\documentclass[11pt, twoside, leqno]{article}
\usepackage{latexsym}
\usepackage{amsmath}
\usepackage{amsfonts}
\usepackage{amssymb}
\usepackage{amsthm}
\usepackage{multido}
\usepackage{color}
\numberwithin{equation}{section}

\theoremstyle{plain}
\newtheorem{theorem}{\indent Theorem}[section]
\newtheorem{proposition}[theorem]{\indent Proposition}
\newtheorem{lemma}[theorem]{\indent Lemma}
\newtheorem{corollary}[theorem]{\indent Corollary}

\theoremstyle{definition}
\newtheorem{definition}[theorem]{\indent Definition}

\newtheorem{remark}[theorem]{\indent Remark}
\newtheorem{stat}[theorem]{Statement}

     \newcommand{\cl}{{\cal L}}
\newcommand{\ch}{{\cal H}}

\newcommand{\ca}{{\cal A}}     \newcommand{\cm}{{\cal M}}
\newcommand{\cb}{{\cal B}}     \newcommand{\cc}{{\cal C}}
     
\newcommand{\cu}{{\cal U}}     

\newcommand{\up}{\raisebox{0.7mm}{$\upharpoonright \,$}}%
\newcommand{\mult}{{\scriptstyle \Box}}
\newcommand{\al}{{\alpha}}     \newcommand{\be}{{\beta}}
     
     \newcommand{\la}{{\lambda}}
    
\newcommand{\ve}{{\varepsilon}}

     \newcommand{\La}{{\varLambda}}
    \newcommand{\vF}{{\varPhi}}
  \newcommand{\1}{\mathit{1}}
  \newcommand{\C}{{\mathbb C}}
\newcommand{\N}{{\mathbb N}}  \newcommand{\R}{{\mathbb R}}
  
\def\x{\relax\ifmmode {\mbox{*}}\else*\fi}
\newcommand{\fd}{\mathcal{D}}\newcommand{\ff}{\mathfrak{F}}
\newcommand{\fm}{\mathfrak{M}}
\newcommand{\lto} {{\longrightarrow}}

\newcommand{\ol}{\overline}
\newcommand{\norm}{{\| \cdot \|}}
\newcommand{\cs}{{C^*}}
\newcommand{\A}{\ca}
\newcommand{\bfn}{\mathbf{n}}
\newcommand{\DD}{\mathcal D}
\newcommand{\BB}{\mathcal B}
\newcommand{\LL}{\mathcal L}
\newcommand{\ip}[2]{\langle {#1}|{#2}\rangle}

\begin{document}
\pagestyle{myheadings} \markboth{\centerline{\small{\sc
            F.\;Bagarello, M.\;Fragoulopoulou, A.\;Inoue and C.\;Trapani}}}
         {\centerline{\small{\sc Locally convex quasi $C^*$-normed algebras}}}
 \title{\bf Locally convex quasi $C^*$-normed algebras}
 \footnotetext{The second author thankfully acknowledges partial support of this research by the
Special Research Account: Grant Nr 70/4/5645, University of Athens.}

\author{\bf F.\;Bagarello, M.\;Fragoulopoulou, A.\;Inoue and C.\;Trapani}

\date{}

\maketitle
\begin{abstract}
If $\ca_0[\| \cdot \|_0]$ is a $\cs$-normed algebra and $\tau$ a
locally convex topology on $\ca_0$ making its multiplication
separately continuous, then $\widetilde{\ca_0}[\tau]$ (completion of
$\ca_0[\tau]$) is a locally convex quasi $*$-algebra over $\ca_0$,
but it is not necessarily a locally convex quasi $*$-algebra over
the $\cs$-algebra $\widetilde{\ca_0}[\| \cdot \|_0]$ (completion of
$\ca_0[\| \cdot \|_0]$). In this article, stimulated by physical examples,
we introduce the notion of
a locally convex quasi $\cs$-normed algebra, aiming at the
investigation of $\widetilde{\ca_0}[\tau]$; in particular, we study
its structure, $*$-representation theory and functional calculus.
\end{abstract}
\setcounter{section} {0}
\section{Introduction}

In the present paper we continue the study introduced in
\cite{ba-maria-i-t} and carried over in \cite{maria-i-kue} and
\cite{bmit}. At this stage, it concerns the investigation of the
structure of the completion of a $\cs$-normed algebra $\ca_0[\|
\cdot \|_0]$, under a locally convex topology $\tau$
$``$compatible$"$ to $\| \cdot \|_0$, that makes the multiplication
of $\ca_0$ separately continuous. The case when $\ca_0[\| \cdot
\|_0]$ is a $\cs$-algebra and $\tau$ makes the multiplication
jointly continuous was considered in \cite{ba-maria-i-t,
maria-i-kue}, while the analogue case corresponding to  separately
continuous multiplication was discussed in \cite{bmit}, where the
so-called \emph{locally convex quasi $\cs$-algebras} were
introduced. In this work, prompted by examples that one meets in
physics, we introduce the notion of \emph{locally convex quasi
$\cs$-normed algebras}, which is wider than that of locally convex
quasi $\cs$-algebras, starting
 with a $\cs$-normed algebra $\ca_0[\| \cdot \|_0]$ and a locally convex topology
 $\tau$,
$``$compatible$"$ to $\| \cdot \|_0$, making the multiplication of
$\ca_0$ separately continuous. For example, let $\cm_0$ be a
$\cs$-normed algebra of operators on a Hilbert space $\ch$, endowed
with the operator norm $\| \cdot \|_0$, $\fd$ a dense subspace of
$\ch$ such that $\cm_0 \fd \subset \fd$ and $\tau_{s^*}$ the
strong$^*$-topology on $\cm_0$ defined by $\fd$. Then, the
$\cs$-algebra $\widetilde{\cm_0}[\| \cdot\|_0]$ does not leave $\fd$
invariant, in general, and so the multiplication $ax$ of $a \in
\widetilde{\cm_0}[\tau_{s^*}]$ and $x \in \widetilde{\cm_0}[\|
\cdot\|_0]$ is not necessarily well-defined, therefore
$\widetilde{\cm_0}[\tau_{s^*}]$ is not a locally convex quasi
$\cs$-algebra over the $\cs$-algebra $\widetilde{\cm_0}[\|
\cdot\|_0]$. Hence, it is meaningful to study not only locally
convex quasi $\cs$-algebras, but also locally convex quasi
$\cs$-normed algebras.

For locally convex quasi $``$$\cs$-normed algebras$"$ we obtain
analogous results to those in \cite{bmit} for locally convex quasi
$``$$\cs$-algebras$"$ despite of the lack of completion and of
weakening the condition $\rm (T_3)$ of \cite{bmit}.

In Section 3 we consider a $\cs$-algebra $\ca_0[\| \cdot \|_0]$ with
a $``$regular$"$ locally convex topology $\tau$ and show that every
unital pseudo-complete symmetric locally convex $*$-algebra
$\ca[\tau]$ such that $\ca_0[\| \cdot \|_0] \subset \ca[\tau]
\subset \widetilde{\ca_0}[\tau]$ is a $GB^*$-algebra over the unit
ball $\cu(\ca_0)$ of $\ca_0[\| \cdot \|_0]$. The latter algebras
have been defined by G.R. Allan \cite{allan} and P.G. Dixon
\cite{dixon} and play an essential role in the unbounded
$*$-representation theory. In Section 4 we define the notion of
locally convex quasi $\cs$-normed algebras and study their general
theory, while in Section 5 we investigate the structure of
commutative locally convex quasi $\cs$-normed algebras. In the final
Section 6 we present locally convex quasi $\cs$-normed algebras of
operators and then we study questions on the $*$-representation
theory of locally convex quasi $\cs$-normed algebras and functional
calculus for the
 $``$commutatively quasi-positive$"$ elements of
 $\widetilde{\ca_0}[\tau]$.

Topological quasi $*$-algebras were introduced in 1981 by G. Lassner
\cite{lass1, lass2}, for facing solutions of certain problems in
quantum statistics and quantum dynamics. But only later (see
\cite[p. 90]{schmudgen}) the initial definition was reformulated in
 the right way, having thus included many more interesting examples.
 Quasi $*$-algebras came in light in 1988 (see \cite{trapani1}, as
 well as \cite{trapani2, ba-t1, ba-t2}), serving as important
 examples of partial $*$-algebras initiated by J.-P. Antoine and W.
 Karwowski in \cite{antoine-karwo1, antoine-karwo2}. A lot of works
 have been done on this topic, which can be found in the treatise
 \cite{ait3}, where the reader will also find a relevant rich
 literature. Partial $*$-algebras and quasi $*$-algebras keep a very
 prominent place in the study of unbounded operators, where the
 latter are the foundation stones for mathematical physics and
 quantum field theory (see, for instance, \cite{ait3, inoue, bag_rev,
 trapani2}).

 Our motivation for such studies comes, on the one hand, from the preceding
 discussion and the promising contribution of the powerful tool
 that the $\cs$-property offers to such studies and, on the other hand, from the
 physical examples of locally convex quasi $C^*$-normed algebras in
 ``dynamics of the BCS-Bogolubov model" \cite{lass2} that will be shortly discussed in Section 7.

\setcounter{section} {1}
\section{Preliminaries}
Throughout the whole paper we consider complex algebras and we
suppose that all topological spaces are Hausdorff. If an algebra
$\ca$ has an identity element, this will be denoted by $\mathit{1}$,
and an algebra $\ca$ with identity $\mathit{1}$ will be called
\emph{unital}.

Let $\ca_0[\| \cdot \|_0]$ be a $\cs$-normed algebra. The symbol $\|
\cdot \|_0$ of the $\cs$-norm will also denote the corresponding
topology. Let $\tau$ be a topology on $\ca_0$ such that
$\ca_0[\tau]$ is a locally convex $*$-algebra. The topologies
$\tau$, $\| \cdot \|_0$ on $\ca_0$ are called \textit{compatible},
whenever for any  Cauchy net $\{x_\alpha\}$ in $\ca_0[\| \cdot \|_0]$ such that
$x_\alpha\rightarrow 0$ in $\tau$, $x_\alpha\rightarrow 0$ in $\| \cdot \|_0$
\cite{bmit}. The completion of $\ca_0$ with respect to $\tau$ will
be denoted by $\widetilde{\ca_0}[\tau]$. In the sequel, we shall
call a directed family of seminorms that defines a locally convex
topology $\tau$, \emph{a defining family of seminorms}.

A \textit{partial $*$-algebra} is a vector space $\ca$ equipped with
a vector space involution $*: \ca \rightarrow \ca: x \mapsto x^*$
and a partial multiplication defined on a set $\Gamma \subset \ca
\times \ca$ such that:

(i) $(x,y) \in \Gamma$ implies $(y^*, x^*) \in \Gamma$;

(ii) $(x, y_1), (x, y_2) \in \Gamma$ and $\lambda, \mu \in
\mathbb{C}$ imply $(x, \lambda y_1+\mu y_2) \in \Gamma$;

(iii) for every $(x,y) \in \Gamma$, a product $xy \in \ca$ is
defined, such that $xy$ depends linearly on $x$ and $y$ and satisfies  the
equality $(xy)^* = y^*x^*$.

Given a pair $(x,y) \in \Gamma$, we say that $x$ is a \textit{left
multiplier} of $y$ and $y$ is a \textit{right multiplier} of $x$,

Quasi $*$-algebras are essential examples of partial $*$-algebras.
If $\ca$ is a vector space and $\ca_0$ a subspace of $\ca$, which is
also a $*$-algebra, then $\ca$ is said to be a \textit{ quasi
$*$-algebra over $\ca_0$} whenever:

(i)$'$ The multiplication of $\ca_0$ is extended on $\ca$ as
follows: The correspondences

$\ca \times \ca_0 \rightarrow \ca : (a,x) \mapsto ax \text{ (left
multiplication of $x$ by $a$) } $ and

$ \ca_0 \times \ca \rightarrow \ca: (x,a) \mapsto xa
\text{ (right multiplication of $x$ by $a$) }$\\
are always defined and are bilinear;

(ii)$'$ $x_1(x_2 a) = (x_1 x_2) a, (a x_1)x_2= a(x_1 x_2)$ and $x_1
(a x_2) = (x_1 a ) x_2$, for all $x_1, x_2 \in \ca_0$ and $a \in
\ca$;

(iii)$'$ the involution $*$ of $\ca_0$ is extended on $\ca$, denoted
also by $*$, such that $(ax)^* = x^* a^*$ and $(xa)^* = a^* x^*$,
for all $x \in \ca_0$ and $a \in \ca$.

For further information cf. \cite{ait3}. If $\ca_0[\tau]$ is a
locally convex $*$-algebra, with separately continuous
multiplication, its completion $\widetilde{\ca_0}[\tau]$ is a quasi
$*$-algebra over $\ca_0$ with respect to the operations:

$\bullet$ \ \ $ax := \displaystyle \lim_\alpha x_\alpha x \text{
(left multiplication) }$, $x \in \ca_0$, $a \in
\widetilde{\ca_0}[\tau]$,

$\bullet$ \ \ $xa := \displaystyle \lim_\alpha x x_\alpha \text{
(right multiplication) }$, $x \in
\ca_0$, $a \in \widetilde{\ca_0}[\tau]$, \\
where $\{x_\alpha\}_{\alpha \in \Sigma}$ is a net in $\ca_0$ such
that $a = \tau$-$\displaystyle \lim_\alpha x_\alpha$.

$\bullet$ \ \ An involution on $\widetilde{\ca_0}[\tau]$ like in
(iii)$'$ is the continuous extension of the involution on $\ca_0$.

A $*$-invariant subspace $\ca$ of $\widetilde{\ca_0}[\tau]$
containing $\ca_0$ is called a \textit{quasi $*$-subalgebra} of
$\widetilde{\ca_0}[\tau]$ if $ax$, $xa$ belong to $\ca$ for any $x
\in \ca_0$, $a \in \ca$. One easily shows that $\ca$ is a quasi
$*$-algebra over $\ca_0$. Moreover, $\ca[\tau]$ is a locally convex
space that contains $\ca_0$ as a dense subspace and for every fixed
$x \in \ca_0$, the maps $\ca[\tau] \rightarrow \ca[\tau]$ with $a
\mapsto ax$ and $a \mapsto xa$ are continuous. An algebra of this
kind is called \textit{locally convex quasi $*$-algebra over
$\ca_0$}.

 We denote by $ \cl^\dag(\fd , \ch) $
the set of all (closable) linear operators $X$ such that $ D(X) = {\fd},\; D(X\x) \supseteq {\fd}.$ The set $
\cl^\dag(\fd , \ch) $ is a  partial *-algebra
 with respect to the following operations: the usual sum $X_1 + X_2 $,
the scalar multiplication $\lambda X$, the involution $ X \mapsto X^\dagger = X\x \up {\fd}$ and the \emph{ (weak)}
partial multiplication $X_1 \mult X_2 = {X_1}^\dagger\x X_2$, defined whenever $X_2$ is a weak right multiplier of
$X_1$ (we shall write $X_2 \in R^{\rm w}(X_1)$ or $X_1 \in L^{\rm w}(X_2)$), that is, iff $ X_2 {\fd} \subset
D({X_1}^\dagger\x)$ and  $ X_1\x {\fd} \subset D(X_2\x).$  $\cl^\dag(\fd , \ch)$ is neither associative
nor semiassociative.

\begin{definition}{\rm
Let $\DD$ be a dense subspace of a Hilbert space $\ch$. A
\textit{$*$-representation} $\pi$ of $\A[\tau]$ is a linear map
from $\A$ into $\cl^\dag(\fd , \ch)$ (see beginning of Section 4) with
the following properties:

(i) $\pi$ is a $*$-representation of $\A_0$;

(ii) $\pi(a)^\dag = \pi(a^*), \forall a \in \A$;

(iii) $ \pi(ax)= \pi(a) \Box \pi(x)$ and $\pi(xa)= \pi(x) \Box
\pi(a), \forall a \in \A$ and $x \in \A_0$, where $\Box$ is the
(weak) partial multiplication of $\cl^\dag(\fd , \ch)$ (ibid.) Having
a $*$-representation $\pi$ as before, we write $\DD(\pi)$ in the
place of $\DD$ and $\ch_\pi$ in the place of $\ch$. By a
\textit{$(\tau, \tau_{s^*})$-continuous $*$-representation} $\pi$
of $\A[\tau]$, we clearly mean continuity of $\pi$, when
$\L^\dag(\fd(\pi), \ch_\pi)$ carries the locally convex topology
$\tau_{s^*}$ (see Section 4). }
\end{definition}

In what follows, we shall need the concept of a $GB^*$-algebra
introduced by G.R. Allan \cite{allan} (see also \cite{dixon}), which
we remind here. Let $\ca[\tau]$  be a locally convex $*$-algebra
with identity $\mathit{1}$ and
 let  $\cb^*$ denote   the
collection of all closed, bounded, absolutely convex subsets $B$ of
$\ca[\tau]$ with the properties: $\mathit{1} \in B$, $B^* = B$ and
$B^2 \subset B$. For each $B \in \cb^*$, the linear span  $A[B]$ of
$B$ is a normed $*$-algebra under the Minkowski functional $\| \cdot
\|_B$ of $B$. When $A[B]$ is complete for each $B \in \cb^*$, then
$\ca[\tau]$ is called {\it pseudo-complete}. Every unital
sequentially complete locally convex $*$-algebra is pseudo-complete
\cite[Proposition (2.6)]{allan1}. A unital locally convex
$*$-algebra $\ca[\tau]$ is called \emph{symmetric} (resp. algebraically symmetric) if for every
$x\in \ca$ the element $\mathit{1} + x^*x$ has an Allan-bounded
inverse in $\ca$ \cite[pp. 91,93]{allan} (resp. if $(1+x^*x)$ has inverse in $\ca$).
 A unital symmetric pseudo-complete locally convex $*$-algebra
 $\ca[\tau]$,
 such that $\cb^*$ has a greatest member, say $B_0$, is said to be a
{\it GB$^*$-algebra over } $B_0$. In this case, $A[B_0]$ is a
$\cs$-algebra.

\setcounter{section} {2}
\section{$\cs$-normed algebras with regular locally convex topology}
\noindent

Let $\ca_0 [\norm_0 ]$ be a $\cs$-normed algebra and
$\widetilde{\ca_0} [\norm_0 ]$ the $\cs$-algebra completion of
$\ca_0 [\norm_0 ]$. Consider a locally convex topology $\tau$ on
$\ca_0$ with the following properties:

$(\rm T_1 )$ $\ca_0 [\tau ]$ is a locally convex $*$-algebra with
separately continuous multiplication.

$(\rm T_2 )$ $\tau \preceq \norm_0$, with  $\tau$ and $\norm_0 $
being compatible.

Then, compatibility of $\tau$, $\norm_0$ implies that:

\begin{itemize}
    \item $\ca_0 [\norm_0 ] \hookrightarrow \widetilde{\ca_0} [\norm_0 ]
    \hookrightarrow \widetilde{\ca_0} [\tau ]$;
    \item
    $\widetilde{\ca_0} [\tau ]$ is a locally convex quasi $*$-algebra
    over the $\cs$-normed algebra $\ca_0 [\norm_0 ]$, but it is
    not necessarily a locally convex quasi $*$-algebra over
    the $\cs$-algebra $\widetilde{\ca_0} [\norm_0 ]$, since
    $\widetilde{\ca_0}[\norm_0]$ is not a locally convex
    $*$-algebra under the topology $\tau$.
\end{itemize}

\textbf{Question.} Under which conditions one could have a
well-defined multiplication of elements in $\widetilde{\ca_0} [\tau
]$ with elements in $\widetilde{\ca_0}[\norm_0 ]$?

We consider the case that the locally convex topology $\tau$ defined
by a directed family of seminorms, say $(p_\la )_{\la \in \Lambda}$,
satisfies in addition to the conditions $\rm (T_1 )$ and $\rm (T_2
)$ an extra $``$good$"$ condition for the $\cs$-norm $\norm_0 $,
called \emph{regularity condition}, denoted by ($R$). That is,

$\rm (R)$ $\forall \ \la \in \La , \ \exists \ \la' \in \La$ and
$\gamma_\la > 0$ : $p_\la (xy)\leqq \gamma_\la \| x\|_0 p_{\la'}
(y)$, $\forall$ $x,y\in \ca_0 [\norm_0 ]$.

In this regard, we have the following

\begin{lemma}\label{le3.1}
Suppose $\ca_0 [\norm_0 ]$ is a $\cs$-normed algebra and $\tau$ a
locally convex topology on $\ca_0$ satisfying the conditions $\rm
(T_1), (T_2)$ and the regularity condition $\rm (R)$ for $\norm_0$. Let
$a$ be an arbitrary element in $\widetilde{\ca_0} [\tau]$ and $y$ an
arbitrary element in $\widetilde{\ca_0} [\norm_0 ]$. Then, the left
resp. right multiplication of $a$ with $y$ is defined by
\[
a\cdot y=\tau -\lim\limits_{\al ,n} x_\al y_n \ \text{ resp. } \
y\cdot a =\tau -\lim\limits_{\al ,n}y_n x_\al,
\]
where $\{x_\al\}_{\al \in \Sigma}$ is a net in $\ca_0 [\tau]$
converging to $a$, $\{y_n\}_{n \in \N}$ is a sequence in $
\ca_0[\norm_0 ]$ converging to $y$ and $\forall \ \la \in \La , \
\exists \ \la' \in \La$ and $\gamma_\la > 0$ :
\[
p_\la (a\cdot y)\leqq \gamma_\la \| y\|_0 p_{\la'} (a), p_\la (y.a)\leqq
\gamma_\la \| y\|_0 p_{\la'} (a).
\]
Under this multiplication $\widetilde{\ca_0} [\tau ]$ is a locally
convex quasi $*$-algebra over the $\cs$-algebra $\widetilde{\ca_0}
[\norm_0 ]$.
\end{lemma}

The proof of Lemma 3.1 follows directly from the regularity condition
$\rm (R)$. If $\ca_0 [\tau]$ is a locally convex $*$-algebra with
jointly continuous multiplication and $\tau \preceq \|\cdot\|_0$,
then it satisfies  the regular condition (R) for $\norm_0$.

\begin{lemma}\label{le3.2} Let $\ca_0 [\norm_0 ]$ be a
$\cs$-normed algebra and $\ca_0 [\tau]$ an $m^*$-convex algebra satisfying
conditions $\rm (T_2)$ and  $\rm (R)$. If $(p_\lambda)_{\lambda \in
\Lambda}$ is a defining family of $m^*$-seminorms for $\tau$ \emph{(i.e.,
submultiplicative $*$-preserving seminorms)} and there is $\lambda_0
\in \Lambda$ such that $p_{\lambda_0 }$ is a norm, then $\tau \sim
\|\cdot\|_0$, where $\sim$ means equivalence of the respective
topologies. In particular, if $\ca_0 [\norm]$ is a normed
$*$-algebra such that $ \norm \leq \norm_0$ and $ \norm $, $\norm_0$
are compatible, then $ \norm \sim \norm_0$.
\end{lemma}

\begin{proof} By $\rm (T_2)$ and  $\rm (R)$ we have
$\widetilde{\ca_0} [\norm_0 ]
    \hookrightarrow \widetilde{\ca_0} [\tau ] \hookrightarrow \widetilde{\ca_0}
    [p_{\lambda_0 }]$,
which by the basic theory of $C^*$-algebras (see e.g.,
\cite[Proposition 5.3]{tak}) implies that $\|x\|_0 \leq p_{\lambda_0
}(x)$, for all $x \in \ca_0 $. Hence, $\tau \sim \|\cdot\|_0$.
\end{proof}

By Lemma 3.2 there does not exist any normed $*$-algebra containing
the $C^*$-algebra $\widetilde{\ca_0} [\norm_0 ]$ properly and
densely.

We now consider whether a $GB^*$-algebra over the unit ball
$\cu(\widetilde{\ca_0} [\norm_0 ])$ exists in $\widetilde{\ca_0}
[\tau ]$. If $\widetilde{\ca_0} [\tau ]$ has jointly continuous
multiplication and $\cu(\widetilde{\ca_0} [\norm_0 ])$ is
$\tau$-closed {\color{red} in $\widetilde{\ca_0} [\tau ]$} then $\widetilde{\ca_0} [\tau ]$ is a $GB^*$-algebra
over $\cu(\widetilde{\ca_0} [\norm_0 ])$, (cf. \cite[Theorem
2.1]{maria-i-kue}).

\begin{theorem}\label{Th3.3} Let $\ca_0 [\norm_0]$ be a unital $\cs$-normed algebra
and $\ca_0 [\tau ]$ a locally convex $*$-algebra such that $\tau$
satisfies the conditions $\rm (T_1), (T_2)$, the regularity condition
$\rm (R)$ for $\norm_0$ and makes the unit ball
$\cu(\widetilde{\ca_0} [\norm_0 ])$ $\tau$-closed. Then every
algebraically symmetric locally convex $*$-algebra $\ca[\tau]$ such
that $\widetilde{\ca_0} [\norm_0] \subset \ca[\tau] \subset
\widetilde{\ca_0} [\tau ]$ is a $GB^*$-algebra over
$\cu(\widetilde{\ca_0}[\|\cdot\|_0])$.
\end{theorem}

\begin{proof}The proof can be done in a similar way to that of
\cite[Theorem 2.2]{ba-maria-i-t}. Here we give a simpler proof.
Without loss of generality we may assume that $\ca_0 [\norm_0]$ is a
$C^*$-algebra. Then we have (see, e.g., proof of \cite[Lemma
2.1]{ba-maria-i-t}):

$\bf (1)$ $(\mathit{1} + a^*a)^{-1} \in \cu(\ca_0 ), \ \forall a \in
\ca$.\\
Moreover, we show that

$\bf (2)$ $\cu(\ca_0)$ is the largest member in $\cb^*(\ca)$.

It is clear that $\cu(\ca_0) \in \cb^*(\ca)$. Suppose now that $B$
is an arbitrary element in $\cb^*(\ca)$ and take $a=a^*$ in $B$. Let
$\cc(a)$ be the maximal commutative $*$-subalgebra of $\ca$
containing $a$ and
\[
\cc_1 \equiv (\cu(\ca_0)\cap \cc(a))\cdot (B\cap \cc(a)).
\]
Then, clearly $\cc_1^*= \cc_1$; by the regular condition $\rm (R)$
$\cc_1$ is $\tau$-bounded in $\cc(a)$, while by the commutativity of
$\cu(\ca_0)\cap \cc(a)$ and $B\cap \cc(a)$ one has that $\cc_1^2
\subset \cc_1$. It is now easily seen that $\overline{\cc_1}^\tau
\in \cb^*(\cc(a))$, where $\cb^*(\cc(a)) = \{B\cap\cc(a) : B \in
\cb^*(\ca)\}$. Thus, there is $B_1 \in \cb^*(\ca)$ such that
$\overline{\cc_1}^\tau = B_1 \cap \cc(a)$.

Since $\cc(a)$ is commutative and pseudo-complete,
$\cb^*(\cc(a))$ is directed \cite[Theorem (2.10)]{allan1}. So for
each $B\in \cb^*(\ca))$ there is $B_1 \in \cb^*(\ca))$ such that
\[
(B\cup\cu(\ca_0))\cap \cc(a) \subset B_1 \cap \cc(a). \ \text{Hence}
\ \ \ca_0 \cap \cc(a) \subset A[B_1] \cap \cc(a),
\]
where $\ca_0 \cap \cc(a)$ is a $\cs$-algebra and $A[B_1] \cap
\cc(a)$ a normed $*$-algebra. An application of Lemma 3.2 gives
\begin{align}
\|x\|_0 = \|x\|_{B_1}, \ \forall x \in \ca_0 \cap \cc(a).
\end{align}
Furthermore, it follows from {\bf (1)} that
$x(\mathit{1}+\frac{1}{n}x^* x)^{-1} \in \ca_0$. Thus,
\[
\|x(\mathit{1}+\frac{1}{n}x^* x)^{-1} - x\|_{B_1} \leq \frac{1}{n}
\|xx^*x\|_{B_1}, \ \forall x \in A[B_1]\cap \cc(a), n \in \N,
\]
which implies that $\ca_0 \cap \cc(a)$ is $\|\cdot\|_{B_1}$-dense in
$A[B_1] \cap \cc(a)$. Therefore, from (3.1) and the fact that $\ca_0
\cap \cc(a)$ is a $\cs$-algebra we get $\ca_0 \cap \cc(a)= A[B_1]
\cap \cc(a)$. It follows that $B \cap \cc(a) \subset B_1 \cap \cc(a)
= \cu(\ca_0) \cap\cc(a)$, from which we conclude
\begin{align}
a \in \cu(\ca_0), \ \forall a \in B, \ \text{ with } \ a^* = a.
\end{align}
Now taking an arbitrary $a \in B$ we clearly have $a^*a \in B$,
hence from (3.2) $a^*a \in \cu(\ca_0)$, which gives $a \in
\cu(\ca_0)$. So, $B \subset \cu(\ca_0)$ and the proof of {\bf (2)}
is complete. Now, since $\cu(\ca_0)$ is the greatest member in
$\cb^*(\ca)$, we have that $A[\cu(\ca_0)]$ coincides with the
$C^*$-algebra $\ca_0$, therefore it is complete.
So \cite[Proposition 2.7]{allan1} implies that
$\ca[\tau]$ is pseudo-complete, hence a
$GB^*$-algebra over $\cu (\ca_0)$.
\end{proof}

\setcounter{section} {3}
\section{Locally convex quasi $\cs$-normed algebras }
\noindent

Let $\ca_0[\norm_0]$ be a $\cs$-normed algebra and $\tau$ a locally
convex topology on $\ca_0$ with $\{p_\la\}_{\la \in \La}$ a defining
family of seminorms. Suppose that $\tau$ satisfies the properties
$\rm (T_1), (T_2)$. The \emph{regularity condition} $\rm (R)$,
considered in the previous Section 2, for $\norm_0$, is too strong
(see Section 6). So in the present Section we weaken this condition,
and we use it together with the conditions $\rm (T_1), (T_2)$, in order to
investigate the locally convex quasi $*$-algebra
$\widetilde{\ca_0}[\tau]$. \emph{The weakened condition $\rm (R)$ will be
denoted by $\rm (T_3)$ and it will read as follows}:

$\rm (T_3)$ $\forall \ \la\in \La, \ \exists \ \la'\in\La$ and $
\gamma_\la > 0 : p_\la(xy) \leqq \gamma_\la \|x\|_0 p_{\la'}(y)$,
for all $ x,y \in \ca_0$ with

\hspace{0.7cm} $xy=yx$.

Then, we first consider the question stated in Section 3, just
before Lemma 3.1, concerning a well-defined multiplication between
elements of $\widetilde{\ca_0} [\tau ]$ and
$\widetilde{\ca_0}[\norm_0 ]$.

If $\ca_0[\norm_0]$ is commutative and $\tau$ satisfies the
conditions $\rm (T_1)-(T_3)$, then $\tau$ fulfils clearly the
regularity condition $\rm (R)$ for $\|\cdot\|_0$, and so by Lemma 3.1,
for arbitrary $a\in \widetilde{\ca_0} [\tau ]$ and $y\in
\widetilde{\ca_0} [\norm_0]$ the left and right multiplications
$a\cdot y$ and $y\cdot a$ are defined, respectively, and $\widetilde{\ca_0}
[\tau ]$ is a locally convex quasi $*$-algebra over the
$\cs$-algebra $\widetilde{\ca_0} [\norm_0]$.

We consider now the afore-mentioned question in the noncommutative
case; for this we set the following

\begin{definition}\label{de4.1} Let $a \in \widetilde{\ca_0} [\tau]$ and
$y \in \widetilde{\ca_0} [\norm_0 ]$. We shall say that \textit{$y$
commutes strongly} with $a$ if there is a net $\{ x_\al\}_{\al \in
\Sigma}$ in $\widetilde{\ca_0}[\norm_0 ]$ such that $x_a
\underset{\tau}{\longrightarrow}a$ and $x_\al y=yx_\al$, for every
$\al \in \Sigma$.
\end{definition}

$\bullet$ In the rest of the paper,
$\widetilde{\ca_0}[\|\cdot\|_0]^\sim [\tau]$, denotes the completion
of the $C^*$-algebra $\widetilde{\ca_0}[\|\cdot\|_0]$ with respect
to the locally convex topology $\tau$. As a set it clearly coincides
with $\widetilde{\ca_0} [\tau ]$, but there are cases that we need
to distinguish them (see Remark 4.6).

\begin{remark}\label{re4.2} Let $a\in \widetilde{\ca_0} [\tau]$ and
$y \in \widetilde{\ca_0} [\norm_0 ]$. Whenever $y \in \ca_0$, the
multiplications $ay$ and $ya$ are always defined by
\[
ay= \lim_\al x_\al y  \ \text{ and } \ ya = \lim_\al y x_\al,
\]
where $\{ x_\al\}_{\al \in \Sigma}$ is a net in $\ca_0$ converging
to $a$ with respect to $\tau$. Hence, we may define the notion $y$
\emph{commutes with} $a$, as usually, i.e., when $ay=ya$. But, even
if $y$ \emph{commutes with} $a$, one has, in general, that $y$
\emph{does not commute strongly with} $a$. Thus, the notion of
strong commutativity is clearly stronger than that of commutativity.
\end{remark}

\begin{lemma}\label{le4.3} Let $\ca_0 [\norm_0 ]$ be a
$\cs$-normed algebra and $\tau$ a locally convex topology on $\ca_0$
that satisfies the properties $\rm (T_1)-(T_3)$. Let $a \in
\widetilde{\ca_0} [\tau]$ and $y \in \widetilde{\ca_0} [\norm_0 ]$
be strongly commuting. Then the multiplications $a \cdot y$ resp.
$y\cdot a$ are defined by
\[
    a\cdot y    =\tau -\lim_\al x_\al y \ \text{ resp. } \
     y\cdot a =\tau -\lim_{\al}yx_\al
     \ \text{ and } \ a\cdot y =y\cdot a,
\]
where $\{ x_\al\}_{\al \in \Sigma}$ is a net in
$\widetilde{\ca_0}[\|\cdot\|_0]$, $\tau$-converging to $a$ and
commutating with $y$. The preceding multiplications provide an
extension of the multiplication of $\ca_0$. Moreover, an analogous
condition to $\rm (T_3)$ holds for the elements $a,y$, i.e.,

$ \rm (T'_3)$ \ $\forall \ \la \in \La \ \exists \ \la' \in \La  \
\text{ and } \ \gamma_\la > 0 : p_\la (a\cdot y)\leqq \gamma_\la  \|
y\|_0 p_{\la'} (a)$.
\end{lemma}

\begin{proof}  Existence of the $\tau -\lim\limits_\al x_\al
    y$ in $\widetilde{\ca_0} [\tau ]$:

     Note that $\{ x_\al y\}_{\al \in \Sigma}$ is a
    $\tau$-Cauchy net in $\widetilde{\ca_0} [\norm_0 ]$.
    Indeed, from $\rm (T_3)$, for every $\la \in \La$, there are $\la' \in
    \La$  and $\gamma_\la > 0$ such that
\[
   p_\la (x_\al
    y-x_{\al'}y)=p_\la \left( (x_\al -x_{\al'})y\right)
    \leqq \gamma_\la \| y\|_0 p_{\la'} (x_\al -x_{\al'}) \underset{\al
    ,\al'}{\longrightarrow}0.
\]
Hence,
    $\tau -\lim\limits_\al x_\al y$ exists in $\widetilde{\ca_0}[\|\cdot\|_0]^\sim [\tau]$,
    which, as already noticed, as a set clearly coincides with $\widetilde{\ca_0} [\tau ]$.

    The existence of the $\tau -\lim\limits_{\al}yx_\al$ in
    $\widetilde{\ca_0}[\|\cdot\|_0]^\sim [\tau]$ is similarly shown and clearly $\tau
    -\lim\limits_\al yx_\al =\tau -\lim\limits_\al x_\al y$.
\medskip

    Independence of $\tau -\lim\limits_\al x_\al y$ from the
    choice of the net $\{ x_\al \}_{\al \in \Sigma}$:

    Let $\{ x'_\be\}_{\be \in \Sigma'}$ be another net in $\ca_0$ such that
    $x'_\be \overset{\tau}{\longrightarrow}a$ and $x'_\be y=yx'_\be$, for all $\be \in \Sigma'$.
    Then,
\[
 x_\al -x'_\be \overset{\tau}{\longrightarrow}0 \ \text{ with }
\    (x_\al -x'_\be )y=y(x_\al -x'_\be ), \forall \ (\al ,\be) \in
\Sigma \times \Sigma'.
\]
Moreover, by $\rm (T_3 )$, for every $\la\in \La$, there exist $\la'
\in \La $ and $\gamma_\la > 0$ such that
\[
 p_\la \left( (x_\al -x'_\be )y\right) \leqq
\gamma_\la \| y\|_0
    p_{\la'}(x_\al -x'_\be ) \underset{\al ,\be}{\rightarrow}0 ;
\]
    this completes the proof of our claim. Thus, we set
\[
a\cdot y := \tau-\lim x_\al y, \ \text{ resp.} \ y \cdot a :=
\tau-\lim y x_\al;
\]
this clearly implies $a \cdot y=y \cdot a.$ Furthermore, using again
$\rm (T_3)$ we conclude that
\[
\forall \ \la \in \La \ \exists \ \la' \in \La \text{ and } \
\gamma_\la > 0 : p_\la (a\cdot y)\leqq \gamma_\la \| y\|_0 p_{\la'}
(a), \ \forall \ a \in \widetilde{\ca_0} [\tau] \ \text{ and } \  y
\in \widetilde{\ca_0} [\norm_0 ],
\]
and this proves $\rm (T'_3)$.
\end{proof}

Now, following \cite{bmit} we define notions of positivity for the
elements of $\widetilde{\ca_0}[\tau]$.

\begin{definition} \label{de4.4}
Let $a \in \widetilde{\ca_0}[\tau]$. Consider the set
\[
 (\ca _0 )_+ :=  \{ x \in \ca_0 : x^* =x \ \text{ and } \
 sp_{\ca _0} (x)\subseteq [0,\infty )  \},
  \]
  where $sp_{\ca _0} (x)$ means spectrum of $x$ in $\ca_0$.
Clearly $(\ca _0 )_+$ is contained in the positive cone of the
$\cs$-algebra $\widetilde{\ca _0} [\norm_0 ]$. The element $a$ is
called \textit{quasi-positive} if there is a net $\{ x_\al
    \}_{\al \in \Sigma}$ in $(\ca _0 )_+$ such that $x_\al
    \underset{\tau}{\longrightarrow}a$.
    In particular, $a$ is called \textit{commutatively quasi-positive} if
    there is a commuting net $\{ x_\al \}_{\al \in \Sigma}$ in $(\ca _0 )_+$ such that $x_\al
    \underset{\tau}{\longrightarrow}a$ .

    Denote by $\widetilde{\ca_0} [\tau]_{q+}$ \emph{the set of all
    quasi-positive elements of} $\widetilde{\ca_0} [\tau]$
    and by $\widetilde{\ca_0} [\tau]_{cq+}$ \emph{the set of all commutatively
    quasi-positive elements of} $\widetilde{\ca_0} [\tau]$.
\end{definition}

An easy consequence of Definition \ref{de4.4} is the following

\begin{lemma}\label{le4.5} $ $
\begin{itemize}
    \item[$(1)$] $\begin{array}{ccc}(\ca_0 )_+& \subset &\widetilde{\ca_0}[\tau ]_{cq+}\\
    \cap &  & \cap\\ \ol{(\ca_0 )_+}^{\norm_0} =\widetilde{\ca_0} [\norm_0 ]_+& \subset &
    \widetilde{\ca_0}  [\tau ]_{q+} \ . \end{array}$
    \item[$(2)$] $\widetilde{\ca_0} [\tau ]_{q+}$ is a positive wedge,
    but it is not necessarily a positive cone. $\widetilde{\ca_0}
    [\tau]_{cq+}$ is not even a positive wedge, in general.
\end{itemize}
\end{lemma}

\begin{remark}\label{re4.6}
As we have mentioned before, the equality
$\widetilde{\ca_0}[\|\cdot\|_0]^\sim [\tau] = \widetilde{\ca_0}
[\tau ]$ holds set-theoretically. We consider the following
notation:
\begin{align*}
&\widetilde{\ca_0}[\|\cdot\|_0]^\sim [\tau]_{q+} \equiv \left\{ a\in
\widetilde{\ca_0} [\tau ] : \exists \ \text{ a net } \ \{ x_\al
\}_{\al \in \Sigma} \ \text{ in } \ \widetilde{\ca_0} [\norm_0 ]_+ :
x_\al \underset{\tau}{\longrightarrow} a\right\} \\
\widetilde{\ca_0}[\|\cdot &\|_0]^\sim [\tau]_{cq+}\equiv \left\{ a
\in \widetilde{\ca_0} [\tau ]: \exists \ \text{ a commuting net } \
\{ x_\al \}_{\al \in \Sigma} \ \text{ in } \ \widetilde{\ca_0}
[\norm_0 ]_+ :  x_\al \underset{\tau}{\longrightarrow}a \right\}.
\end{align*}
Then,
\begin{align}
    \widetilde{\ca_0}[\|\cdot\|_0]^\sim [\tau]_{q+} = \widetilde{\ca_0} [\tau
    ]_{q+}, \ \text{ but} \ \
    \widetilde{\ca_0}[\|\cdot\|_0]^\sim [\tau]_{cq+} \supsetneqq
    \widetilde{\ca_0} [\tau ]_{cq+}, \ \text{ in general}.
\end{align}
 If $\ca_0$ is commutative, then
\[
    \widetilde{\ca_0} [\tau ]_{cq+} =\widetilde{\ca_0}[\|\cdot\|_0]^\sim [\tau]_{cq+}
    =\widetilde{\ca_0}[\|\cdot\|_0]^\sim [\tau]_{q+}
    =\widetilde{\ca_0} [\tau ]_{q+}.
\]
\end{remark}

The following Proposition \ref{pro4.7} plays an important role in
the present paper. It is a generalization of Proposition 3.2 in
\cite{bmit}, stated for locally convex quasi $\cs$-algebras, to the
case of locally convex quasi $\cs$-normed algebras.

\begin{proposition}\label{pro4.7}
Let $\ca_0[\norm_0 ]$ be a unital $\cs$-normed algebra and $\tau$ a
locally convex topology on $\ca_0$ that fulfils the conditions $\rm
(T_1 )-(T_3 )$. Suppose that the next condition $\rm (T_4)$ holds:

$\rm (T_4 ) \ $ The set $\cu (\widetilde{\ca_0} [\norm_0 ])_+ \equiv
\{ x\in \widetilde{\ca_0} [\norm_0 ]_+ : \| x\|_0 \leqq 1\}$ is
$\tau$-closed in {\color{red}$\widetilde{\ca_0} [\tau ]$ (or, equivalently, it is $\tau$-complete)}.

Then, $\widetilde{\ca_0} [\tau]$ is a locally convex quasi
$*$-algebra over $\ca_0$ with the properties:

\emph{(1)} $a\in \widetilde{\ca_0} [\tau ]_{cq+}$ implies that
    $\mathit{1}+a$ is invertible with $(\mathit{1}+a)^{-1}$ in $\cu (\widetilde{\ca_0} [\norm_0 ])_+$.

\emph{(2)} For $a\in \widetilde{\ca_0} [\tau ]_{cq+}$ and $\ve
    >0$, the element $a_\ve := a\cdot (\mathit{1}+\ve a)^{-1}$ is
    well-defined, $a-a_\ve \in \widetilde{\ca_0}[\|\cdot\|_0]^\sim [\tau]_{cq+}$
    and $a=\tau -\lim\limits_{\ve \downarrow 0}
    a_\ve $.

\emph{(3)} $\widetilde{\ca_0} [\tau ]_{cq+}\cap (-\widetilde{\ca_0}
    [\tau]_{cq+})=\{ 0\}$.

\emph{(4)} Furthermore, suppose that the following condition

    $\rm (T_5 ) \ $
    $\widetilde{\ca_0} [\tau]_{q+}\cap \widetilde{\ca_0}
    [\norm_0 ]= \widetilde{\ca_0} [\norm_0 ]_+$ \\
is satisfied. Then, if $a\in \widetilde{\ca_0} [\tau ]_{cq+}$ and
$y\in
    \widetilde{\ca_0} [\norm_0 ]_+$ with $y-a \in \widetilde{\ca_0}
    [\tau]_{q+}$, one has that $a\in \widetilde{\ca_0}
    [\norm_0 ]_+$.
\end{proposition}

\begin{proof} (1) There exists a commuting net $\{ x_\al
\}_{\al \in \Sigma}$ in $(\ca_0 )_+$ with $x_\al
\underset{\tau}{\longrightarrow} a$ and $x_\al x_{\al'}
=x_{\al'}x_\al$, for all $\al, \al' \in \Sigma$. Using properties of
the positive elements in a $\cs$-algebra, and condition $\rm
(T_3 )$ we get that for every $\la \in \La$ there are $\la' \in \La$
and  $\gamma_\la > 0$ such that:
\begin{align*}
    p_\la & \left( (\mathit{1}+x_\al )^{-1}-(\mathit{1}+x_{\al'})^{-1}\right)
    =p_\la \left( (\mathit{1}+x_\al )^{-1} (x_{\al'}-x_\al )(\mathit{1}+x_{\al'}
    )^{-1}\right) \\
    &\leqq \gamma_\la \| (\mathit{1}+x_\al )^{-1} \|_0 \| (\mathit{1}+x_{\al'})^{-1}
    \|_0 p_{\la'} (x_{\al'} -x_{\al}) \leqq \gamma_\la p_{\la'} (x_{\al'}-x_\al
    )\underset{\al ,\al'}{\longrightarrow}0.
\end{align*}
    Hence $\{ (\mathit{1}+x_\al )^{-1}\}_{\al \in \Sigma}$ is a Cauchy net in $\widetilde{\ca_0} [\tau]$
    consisting of elements of $\cu (\widetilde{\ca_0}  [\norm_0
    ])_+$, the latter set being $\tau$-closed by $\rm (T_4 )$. Hence,
    there exists $y\in \cu (\widetilde{\ca_0}  [\norm_0])_+$ such that
\begin{align}
 (\mathit{1}+x_\al )^{-1} \underset{\tau}{\longrightarrow}
 y.
\end{align}
    We shall show that $(\mathit{1}+a)^{-1}$ exists in $\cu (\widetilde{\ca_0}
    [\norm_0 ])_+ $ and coincides with $y$.
    It is easily seen that, for each index $\al \in \Sigma$, $(\mathit{1}+x_\al )^{-1}$
    commutes strongly with $(\mathit{1}+a)$, so that $(\mathit{1}+a)\cdot (\mathit{1}+x_\al )^{-1}$ is
    well-defined (Lemma \ref{le4.3}).
     Similarly, $(x_\al -a)\cdot
    (\mathit{1}+x_\al )^{-1} =\mathit{1} -(\mathit{1}+a)\cdot (\mathit{1}+x_\al )^{-1}$ is well-defined,
    therefore using $\rm (T'_3)$ of Lemma 4.3, we have that for all $\la \in\La$
there are $\la' \in\La$  and $\gamma_\la >0$ with
\begin{align*}
p_\la (\mathit{1}-(\mathit{1}+a)\cdot (\mathit{1}+x_\al )^{-1})=
    p_\la ((x_\al -a)\cdot (\mathit{1}+x_\al )^{-1})\leqq \gamma_\la  p_{\la'}(x_\al
    -a)\underset{\al}{\lto}0.
\end{align*}
Thus,
    $(\mathit{1}+a)\cdot (\mathit{1}+x_\al )^{-1} \underset{\tau}{\lto} \mathit{1}$.
By the above,
\begin{align*}
\mathit{1}+x_\al \underset{\tau}{\lto} \mathit{1}+a \ \text{ and } \
(\mathit{1}+x_\al )y =y(\mathit{1}+x_\al), \forall \ \al \in \Sigma.
\end{align*}
Hence, $y$ commutes strongly with $\mathit{1}+a$, therefore
$(\mathit{1}+a)\cdot
     y$ is well-defined by Lemma \ref{le4.3}. Now, since
$x_\al \underset{\tau}{\lto} a$, we have that
\begin{align}
\forall \ \la\in\La \ \text{ and } \
     \forall \ \ve >0, \exists \ \al_0 \in \Sigma :
     p_{\la}(x_{\al'}-a)<\ve, \ \forall \ \al'\geqq \al_0.
\end{align}
Using $\rm (T_3)$, $\rm (T'_3)$ of Lemma 4.3, and relations
(4.3), (4.2) we obtain
\begin{align*}
        p_\la &((\mathit{1}+a)\cdot (\mathit{1}+x_\al )^{-1}-(\mathit{1}+a)\cdot y) \\
        &\leqq p_\la
        ((\mathit{1}+a)\cdot (\mathit{1}+x_\al )^{-1}-(\mathit{1}+x_{\al_0})(\mathit{1}+x_\al )^{-1})\\
        &+p_\la ((\mathit{1}+x_{\al_0})(\mathit{1}+x_\al )^{-1}-(\mathit{1}+x_{\al_0})y)
        +p_\la ((\mathit{1}+x_{\al_0})y-(\mathit{1}+a)y) \\
        &\leqq \gamma_\la p_{\la'}(a-x_{\al_0})+\gamma_\la \| \mathit{1}+x_{\al_0}\|_0 p_{\la'}
        ((\mathit{1}+x_\al )^{-1}-y)+\gamma_\la p_{\la'}(x_{\al_0}-a)\\
        &<2\ve +\gamma_\la \| \mathit{1}+x_{\al_0} \|_0 p_{\la'}((\mathit{1}+x_\al )^{-1}-y), \ \
        \forall \ \ve >0.
\end{align*}
     Hence,
     \[
        0\leqq \lim_\al p_\la ((\mathit{1}+a)\cdot (\mathit{1}+x_\al )^{-1}-(\mathit{1}+a)\cdot
        y)\leqq 2\ve, \ \forall \ve >0,
     \]
     which implies
     \[
        \lim_\al p_\la ((\mathit{1}+a)\cdot (\mathit{1}+x_\al )^{-1}-(\mathit{1}+a)\cdot y )=0.
     \]
     Consequently,
 \begin{align}
        (\mathit{1}+a)\cdot (\mathit{1}+x_\al )^{-1}\underset{\tau}{\lto}(\mathit{1}+a)\cdot y.
 \end{align}
     Similarly, $(\mathit{1}+x_\al )^{-1}\cdot
     (\mathit{1}+a)\underset{\tau}{\lto}y\cdot (\mathit{1}+a)$. So from (4.3) and (4.4) we
     conclude that $(\mathit{1}+a)\cdot y=y\cdot (\mathit{1}+a)=\mathit{1}$, therefore
      $y=(\mathit{1}+a)^{-1}.$

(2) By (1), for every $ \ve >0$, the element $(\mathit{1}+\ve
a)^{-1}$ exists in $\cu (\widetilde{\ca_0} [\norm_0 ])_+$, and
commutes strongly with $a$. Hence (see Lemma \ref{le4.3}),
$a_\ve:=a\cdot (\mathit{1}+\ve a)^{-1}$ is well-defined. Moreover,
applying $\rm (T'_3 )$ of Lemma 4.3, we have that for all $\la \in
\La$, there exist $\la' \in \La$ and $\gamma_\la >0$ such that
\[
    p_\la (\mathit{1}-(\mathit{1}+\ve a)^{-1})=\ve p_\la (a\cdot (\mathit{1}+\ve a)^{-1}) \leqq
    \ve \gamma_\la \| (\mathit{1}+\ve a)^{-1}\|_0 p_{\la'}(a) \leqq \ve \gamma_\la p_{\la'}(a).
\]
Therefore,
\begin{align}
\tau -\lim_{\ve \downarrow 0} (\mathit{1}+\ve a)^{-1}=\mathit{1}.
\end{align}
On the other hand, since $(\mathit{1}-(\mathit{1}+\ve a)^{-1})$
commutes strongly with $a$ and \linebreak[4] $a_\ve
=\ve^{-1}(\mathit{1}-(\mathit{1}+\ve a)^{-1})$, $\ve>0$, we have
\begin{align}
 (\mathit{1}-(\mathit{1}+\ve a)^{-1})\cdot a=a \cdot (\mathit{1}-(\mathit{1}+\ve a)^{-1})
 =a- a_\ve  \in \widetilde{\ca_0} [\norm_0
 ]^\sim [\tau]_{cq+}.
\end{align}
Using (4.5), (4.6) and same arguments as in (4.4), we get that $\tau
-\lim\limits_{\ve \downarrow 0}a_\ve =a$.

(3) Let $a\in \widetilde{\ca_0} [\tau]_{cq+}\cap (-\widetilde{\ca_0}
[\tau]_{cq+})$ and $\ve >0$ sufficiently small. By (2) (see also
Remark \ref{re4.6}), we have
\begin{align*}
    \widetilde{\ca_0}[\norm_0 ]^\sim [\tau]_{cq+} \ni a\cdot (\mathit{1}+\ve
    a)^{-1}\underset{\tau} {\lto} \ a; \ \text{ in the same way } \
    -a \cdot (\mathit{1}-\ve a)^{-1}\underset{\tau}{\lto}-a.
\end{align*}
Now the element
\begin{align*}
    x_\ve \equiv \ a\cdot (\mathit{1}+\ve a)^{-1}-(-a)\cdot (\mathit{1}-\ve a)^{-1}
        =2a \cdot (\mathit{1}+\ve a)^{-1}(\mathit{1}-\ve a)^{-1}
\end{align*}
belongs to $\widetilde{\ca_0}[\norm_0 ]_+$ by (1) and the functional
calculus of commutative $\cs$-algebras. Similarly,
   $ -x_\ve =2(-a)\cdot (\mathit{1}-\ve a)^{-1}(\mathit{1}+\ve a)^{-1}\in \widetilde{\ca_0}
    [\norm_0 ]_+.$
Hence,
\begin{align*}
    x_\ve \in \widetilde{\ca_0} [\norm_0 ]_+ \cap (-\widetilde{\ca_0}
    [\norm_0 ]_+ )=\{ 0\}, \ \text{ so that } \
    a\cdot (\mathit{1}+\ve a)^{-1}=-a\cdot (\mathit{1}-\ve a)^{-1}.
\end{align*}
Furthermore, by (2),
\begin{align*}
  a=\tau -\lim_{\ve \downarrow 0}a\cdot (\mathit{1}+\ve a)^{-1}=\tau
    -\lim_{\ve \downarrow 0} (-a) \cdot (\mathit{1}-\ve a)^{-1}=-a, \
    \text{ so } \  a=0.
\end{align*}
(4) Note that
   $ y-a_\ve = (y-a)+(a-a_\ve )\in \widetilde{\ca_0} [\tau ]_{q+}$,
since (by (4) and (2) resp.) the elements  $y-a$, $a-a_\ve$ belong
to $\widetilde{\ca_0} [\tau]_{q+}$ and the latter set is a positive
wedge according to Lemma 4.5(2). On the other hand,
\[
a_\ve = a\cdot (\mathit{1}+\ve a)^{-1} = (\mathit{1}+\ve a)^{-1}
\cdot a= \ve^{-1} (\mathit{1}-(\mathit{1}+\ve a)^{-1}) \in
\widetilde{\ca_0}[\|\cdot\|_0].
\]
Thus, taking under consideration the assumption $\rm (T_5)$ we
conclude that
\[
    y-a_\ve \in \widetilde{\ca_0} [\tau]_{q+}\cap \widetilde{\ca_0}
    [\norm_0 ]= \widetilde{\ca_0} [\norm_0 ]_+,
\]
which clearly gives
$    \| a_\ve \|_0 \leqq \| y\|_0$, for every $\ve >0$.
Applying $\rm (T_4 )$, we show that $a\in \widetilde{\ca_0} [\norm_0
]_+ $.
\end{proof}

\begin{definition}\label{de4.8}
Let $\ca_0[\norm_0 ]$ be a unital $\cs$-normed algebra, $\tau$ a
locally convex topology on $\ca_0$ satisfying the conditions $\rm
(T_1 )-(T_5 )$ (for $\rm (T_4), (T_5 )$ see the previous
proposition). Then,

$\bullet$   a quasi $*$-subalgebra $\ca$ of the locally convex quasi
    $*$-algebra
    $\widetilde{\ca_0} [\tau]$ over ${\ca}_0$ containing $\widetilde{\ca_0}
    [\norm_0 ]$ is said to be a \textit{locally convex quasi $\cs$-normed algebra over
    $\ca_0$}.

$\bullet$   A locally convex quasi $\cs$-normed algebra $\ca$ over
    $\ca_0$ is said to be
    \textit{normal} if $a\cdot y\in \ca$ whenever $a\in \ca$ and
    $y\in \widetilde{\ca_0} [\norm_0 ]$ commute strongly.

$\bullet$    A locally convex quasi $\cs$-normed algebra $\ca$ over
    $\ca_0$ is called
    a \textit{locally convex quasi $\cs$-algebra} if $\ca_0
    [\norm_0 ]$ is a $\cs$-algebra.
\end{definition}

\noindent Note that the condition $\rm (T_3)$ in the present paper
is weaker than the condition

$\rm (T_3)$  $\forall \ \la\in \La, \ \exists \ \la'\in\La :
p_\la(xy) \leqq \|x\|_0 p_{\la'}(y), \ \forall \ x,y \in \ca_0 $\
with  $xy=yx$

\noindent in \cite{bmit}. Nevertheless, results for locally convex
quasi $\cs$-algebras in \cite{bmit} are valid in the present paper
for the wider class of locally convex $\cs$-normed algebras. It
follows, by the very definitions, that \emph{a locally convex quasi
$\cs$-algebra is a normal locally convex quasi $\cs$-normed
algebra}. A variety of examples of locally convex quasi
$\cs$-algebras are given in \cite{bmit}, Sections 3 and 4. Examples
of locally convex quasi $\cs$-normed algebras are presented in
Sections 6 and 7

An easy consequence of Definition 4.8 and Lemma 4.3 is the following

\begin{lemma}\label{le4.9}
Let $\ca_0[\|\cdot\|_0]$ and $\tau$ be as in Definition \emph{ \ref{de4.8}}.
Then the following hold:

    \emph{(1)} $\widetilde{\ca_0} [\tau]$ is a normal locally convex
    quasi $\cs$-normed algebra over $\ca_0$.

    \emph{(2)} Suppose $\ca$ is a commutative locally convex quasi
    $\cs$-normed algebra over $\ca_0$. Then $\ca \cdot \widetilde{\ca_0}
    [\norm_0 ]\equiv$ linear span of $\{ a\cdot y : a \in \ca, y \in \widetilde{\ca_0}
    [\norm_0 ]\}$ is a commutative locally convex quasi
    $\cs$-algebra over $\widetilde{\ca_0} [\norm_0 ]$ under the
    multiplication $a\cdot y$ $(a\in \ca, y\in \widetilde{\ca_0}
    [\norm_0 ])$. In particular, if $\ca$ is normal, then $\ca$ is a
    commutative locally convex quasi $\cs$-algebra over $\widetilde{\ca_0}
    [\norm_0 ]$.
\end{lemma}

\setcounter{section} {4}
\section{Commutative locally convex quasi $\cs$-normed algebras}
\noindent

In this Section, we discuss briefly some results on the structure of
a commutative locally convex quasi $\cs$-normed algebra $\ca [\tau]$
and on a functional calculus for its quasi-positive elements, that
are similar to those in \cite[Sections 5 and 6]{bmit}.

 Let $\ca [\tau]$ be a commutative locally convex
quasi $\cs$-normed algebra over $\ca_0$ (see Definition 4.8). Then,
\[
    \ca_0 [\norm_0 ]\subset \widetilde{\ca_0} [\norm_0 ]\subset \ca
    [\tau] \subset \ca [\tau ]\cdot \widetilde{\ca_0} [\norm_0 ]\subset \widetilde{\ca_0}
    [\tau ],
\]
where $\ca_0 [\norm_0 ]$ is a commutative unital $\cs$-normed
algebra and $\ca [\tau ]\cdot \widetilde{\ca_0} [\norm_0 ]$ is a
commutative locally convex quasi $\cs$-algebra over the unital
$\cs$-algebra $\widetilde{\ca_0} [\norm_0 ]$ according to Lemma
4.9(2). Thus, using some results of the Sections 5, 6 in \cite{bmit}
for the latter algebra we obtain information for the structure of
$\ca [\tau]$.

Let $W$ be a compact Hausdorff space, $\C^*=\C\cup\{\infty\}$, and $\ff(W)_+$ a set
of $\C^*$-valued positive continuous functions on $W$, which take
the value $\infty$ on at most a nowhere dense subset $W_0$ of
$W$. The set
\[
    \ff(W) \equiv \{ fg_0 +h_0 : f\in \ff (W )_+ \ \text{ and } \ g_0 ,h_0 \in
    \cc(W )\},
\]
where $\cc(W )$ is the $\cs$-algebra of all continuous $\C$-valued
functions on $W$, is called \emph{the set of $\C^*$-valued
continuous functions on $W$ generated by the wedge $\ff (W)_+$
and the $\cs$-algebra} $\cc (W)$. Using \cite[Definition 5.6]{bmit} and $\ff(W)$
we get the following theorem, which is an application of Theorem
5.8 of \cite{bmit} for the commutative locally convex quasi
$\cs$-algebra $\ca [\tau ]\cdot \widetilde{\ca_0} [\norm_0 ]$ over
the unital commutative $\cs$-algebra $\widetilde{\ca_0} [\norm_0 ]$,
with $\ca [\tau]_{q+} \cdot \widetilde{\ca_0} [\norm_0 ]$, in the
place of $\fm (\ca_0, \ca[\tau]_{q+})$.

\begin{theorem}\label{th5.1}
There exists a map $\vF$ from $\ca [\tau]_{q+} \cdot
\widetilde{\ca_0} [\norm_0 ]$ onto $\ff(W)$, where $W$ is the
compact Hausdorff space corresponding to the Gel'fand space of the
unital commutative $\cs$-algebra $\widetilde{\ca_0} [\norm_0 ]$,
such that:
\begin{itemize}
    \item[(i)] $\vF (\ca [\tau ]_{q+})=\ff (W)_+$ and
    $ \ {\vF}(\la a+b)=\la \vF (a)+\vF (b)$,
    $\forall$ $a,b\in \ca [\tau]_{q+}$,
    $\la \geqq 0$;

    \item[(ii)] $\vF$
    is an isometric $*$-isomorphism from $\widetilde{\ca_0} [\norm_0
    ]$ onto $\cc(W)$;
    \item[(iii)]
    $\vF (ax)=\vF (a)\vF (x)$,
    $\ \vF ((\la a+b)x)=(\la \vF (a)+\vF (b))\vF (x) \ $ and
    $\ \vF (a(x_1 +x_2 ))=\vF (a)(\vF (x_1 )+\vF (x_2 ))$, $\forall
    \ a,b\in \ca [\tau ]_{q+}$, $x, x_1 ,x_2  \in \ca_0$
    and $\la \geqq 0$.
\end{itemize}
\end{theorem}

$\bullet$ Further we consider a functional calculus for the
quasi-positive elements of the commutative locally convex quasi
$\cs$-normed algebra $\ca[\tau]$ over $\ca_0$. For this, we must
extend the multiplication of $\ca[\tau]$.

Let $a,b \in \ca [\tau]_{q+}$. Then (see also \cite
    [Definition 6.1]{bmit}), $a$ is called \textit{left multiplier} of $b$
    if there are nets $\{ x_\al \}_{\al \in \Sigma} ,\{ y_\be\}_{\be \in \Sigma'} $
    in $(\ca_0 )_+$ such that $x_\al
    \underset{\tau}{\lto}a$, $y_\be \underset{\tau}{\lto}b$ and $x_\al
    y_\be \underset{\tau}{\lto} c$, where the latter
    means that the double indexed net $\{x_\al y_\be\}_{(\al,\be) \in \Sigma \times \Sigma'}$
converges to $c \in \ca[\tau]$.
    Then, we set
    \[
    a\cdot b := c = \tau -\lim\limits_{\al, \be}x_\al y_\be,
    \]
where the multiplication $a\cdot b$ is well defined, in the sense
that it is independent of the choice of the nets $\{ x_\al \}_{\al
\in \Sigma} ,\{ y_\be\}_{\be \in \Sigma'} $, as follows from the
proof of Lemma 6.2 in \cite{bmit} applying arguments of the proof of
Proposition 4.7. In the sequel, we simply denote $a\cdot b$ by
$ab$. In analogy to Definition 6.3 of \cite{bmit}, if $x,y \in
\widetilde{\ca_0} [\norm_0 ]$ and $a,b \in
\ca [\tau]_{q+}$ with $a$ left multiplier of $b$, we may define the product of the
elements $ax$ and $by$ as follows:
\[
(ax)(by):=(ab)xy.
\]

The spectrum of an element $a\in \ca[\tau
]_{q+}$, denoted by $\sigma_{\widetilde{\ca_0} [\norm_0 ]}(a)$, is
defined as in Definition 6.4 of \cite{bmit}.

So using Theorem 5.1, it is shown (cf., for instance, Lemma 6.5 in
\cite{bmit}) that for every $a\in \ca[\tau ]_{q+}$, one has that
$\sigma_{\widetilde{\ca_0}[\norm_0 ]}(a)$ \emph{is a locally compact
subset of $\C^*$ and} \linebreak[4] $\sigma_{\widetilde{\ca_0}
[\norm_0 ]}(a) \subset \R_+ \cup \{\infty \}$.

According to the above, and taking into account the comments after Lemma 6.5
in \cite{bmit} with $\widetilde{\ca_0}[\|\cdot\|_0]$ in the place of $\ca_0$,
the next Theorem 5.2
provides a generalization of \cite[Theorem 6.6]{bmit} in the setting
of commutative locally convex quasi $\cs$-normed algebras. In particular, Theorem
5.2 supplies us with a functional calculus for the quasi-positive
elements of the commutative locally convex quasi $\cs$-normed
algebra $\ca [\tau]$.

\begin{theorem}\label{th5.2}
Let $a\in \ca [\tau]_{q+}$. Let $a^n$ be well-defined
for some $n\in \N$. Then there is a unique $*$-isomorphism $f\to
f(a)$ from $\bigcup ^{n} _{k=1} \cc_k
(\sigma_{\widetilde{\ca_0} [\norm_0 ]}(a))$ into $\ca[\tau] \cdot
\widetilde{\ca_0} [\norm_0 ]$ such that:
\begin{itemize}
    \item[(i)] If $u_0 (\la )=1$, with $u_0 \in \bigcup ^{n} _{k=1} \cc_k
(\sigma_{\widetilde{\ca_0} [\norm_0 ]}(a))$ and $\la \in
\sigma_{\widetilde{\ca_0} [\norm_0 }(a)$, then $u_0 (a)=\mathit{1}$.
    \item[(ii)] If $u_1 (\la) =\la$ with $u_1 \in \bigcup ^{n} _{k=1} \cc_k
(\sigma_{\widetilde{\ca_0} [\norm_0 ]}(a))$ and $\la \in
\sigma_{\widetilde{\ca_0} [\norm_0]}(a)$, then $u_1 (a)=a$.
    \item[(iii)] $(\la _1 f_1 +f_2 )(a)=\la_1f_1(a) + f_2(a)$, $\forall$ $f_1, f_2 \in \cc_k
(\sigma_{\widetilde{\ca_0} [\norm_0 ]}(a))$ and $\la_1 \in \C$;
        $(f_1 f_2 )(a)=f_1 (a)f_2 (a)$, $\forall$ $f_j \in
    \cc_{k_j}(\sigma_{\widetilde{\ca_0} [\norm_0 ]}(a))$, $j=1,2$, with $k_1
+k_2 \leqq n$.
    \item[(iv)] Denoting with $\cc_b$ the set of the bounded and continuous functions,
    the map \linebreak[4] $f\to f(a)$ restricted to $\cc_b(\sigma_{\widetilde{\ca_0} [\norm_0 ]}(a))$
    is an  isometric $*$-isomorphism of the $\cs$-algebra
    $\cc_b(\sigma_{\widetilde{\ca_0} [\norm_0 ]}(a))$ on the closed
    $*$-subalgebra of $\widetilde{\ca_0} [\norm_0 ]$ generated by $\mathit{1}$
    and \linebreak[4] $(\mathit{1}+a)^{-1}$.
\end{itemize}
\end{theorem}

Applying Theorem 5.2 and Proposition 4.7 in the proof of
\cite[Corollary 6.7]{bmit} we get the following

\begin{corollary}\label{co5.3}
Let $a\in \ca [\tau ]_{q+}$ and $n\in \N$. Then, there exists unique
$b$ in $\ca [\tau]_{q+}\cdot \widetilde{\ca_0} [\norm_0 ]$ such that
$a=b^n$. The unique element $b$ is called quasi $n$th-root of $a$
and we write $b=a^{\frac{1}{n}}$.
\end{corollary}

\setcounter{section} {5}
\section{Structure of noncommutative locally convex quasi
$\cs$-normed algebras}

\noindent Using the notation of \cite[Section 4]{bmit} (see also
\cite{ait3}), let $\ch$ be a Hilbert space, $\fd$ a dense subspace
of $\ch$ and $\cm_0 [\norm_0]$ a unital $C^*$-normed algebra on
$\ch$, such that
\[
\cm_0 \fd  \subset \fd , \ \text{ but } \ \widetilde{\cm}_0 [\norm_0
] \fd \not\subset \fd.
\]

Then, the restriction $\cm_0\upharpoonright \fd $ of $\cm_0$ to $\fd
$ is an $O^*$-algebra on $\fd $, so that an element $X$ of $\cm_0$
may be regarded as an element $X \upharpoonright \fd $ of $\cm_0
\lceil \fd $. Moreover, let
\[
\cm_0 \subset \cm \subset \cl^\dag(\fd , \ch),
\]
where $\cm$ is an $O^*$-vector space on $\fd $, that is, a
$*$-invariant subspace of $\cl^\dag(\fd , \ch)$. Denote by
$\cb(\cm)$ the set of all bounded subsets of $\fd [t_\cm]$ ($t_\cm$
is the graph topology on $\cm$; see \cite[p.9]{inoue}) and by
$\cb_f(\fd)$ the set of all finite subsets of $\fd  $. Then
$\cb_f(\fd ) \subset \cb(\cm)$ and both of them are admissible  in
the sense of \cite[p. 522]{bmit}.

We recall the topologies $ \tau_{s^*}, \ \tau^{u}_* (\cb), \
\tau^{u}_* (\cm)$ defined in \cite[pp. 522-523]{bmit}. More
precisely, for an arbitrary admissible subset $\cb$ of $\cb(\cm)$,
and any $\fm \in \cb$ consider the following seminorm:
\begin{align*}
p^\fm_\dag (X) := \sup_{\xi \in \fm} \{ \|X \xi \|+ \| X^\dag \xi \|
\}, \quad X \in \cm.
\end{align*}
We call the corresponding locally convex topology on $\cm$ induced
by the preceding family of seminorms, \textit{strongly$^*$
$\cb$-uniform topology } and denote it by $\tau^u_*(\cb)$. In
particular, the strongly$^*$ $\cb(\cm)$-uniform topology will be
simply called \textit{strongly$^*$ $\cm$-uniform topology} and will
be denoted by $\tau^u_*(\cm)$. In Schm\"{u}dgen's book
\cite{schmudgen}, this topology is called \emph{bounded topology}.
The strongly$^*$ $\cb_f(\fd)$-uniform topology is called
\textit{strong$^*$-topology on $\cm$}, denoted by $\tau_{s^*}$. All
three topologies are related in the following way:
\begin{align*}
      \tau_{s^*} \preceq \tau^{u}_* (\cb)\preceq \tau^{u}_* (\cm).
\end{align*}
Then, one gets that
\begin{align}
    \cm_0 [\norm_0 ]\subset \widetilde{\cm}_0 [\norm_0 ]\subset
    \widetilde{\cm}_0 [\tau^{u}_* ]\subset \widetilde{\cm}_0 [\tau_{s^*}]
    \subset \cl^\dagger (\fd, \ch).
\end{align}
In this regard, we have now the following

\begin{proposition}\label{pro6.1}
Let $\cm_0 [\norm_0 ]$, $\cm$ be as before. Let $\cb$ be any
admissible subset of $\cb (\cm)$. Then $\widetilde{\cm}_0[\tau^{u}_*
(\cb )]$ is a locally convex quasi $\cs$-normed algebra over
$\cm_0$, which is contained in $\cl^\dag (\fd ,\ch)$. In particular,
$\widetilde{\cm}_0[\tau_{s^*}]$ is a locally convex quasi
$\cs$-normed algebra over $\cm_0$. Furthermore, if $A\in
\widetilde{\cm}_0[\tau^u _* (\cb )]$ and $Y\in
\widetilde{\cm}_0[\norm_0 ]$ commute strongly, then $A\Box Y$ is
well-defined and
\[
    A\Box Y=A\cdot Y =Y\cdot A =Y\Box A.
\]
\end{proposition}

\begin{proof}
It is easily checked that $\widetilde{\cm}_0[\tau^u _* (\cb )]$ and
$\widetilde{\cm}_0[\tau_{s^*}]$  are locally convex quasi
$\cs$-normed algebras over $\cm_0$. Suppose now that $A\in
\widetilde{\cm}_0[\tau^u _*(\cb) ]$ and $Y\in
\widetilde{\cm}_0[\norm_0 ]$ commute strongly. Then, there is a net
$\{ X_\al \}_{\al \in \Sigma}$ in $\cm_0$ such that $X_\al
Y=YX_\al$, for all $\al \in \Sigma$ and $A= \tau^u _*
(\cb)-\lim\limits_\al X_\al $. Since
\[
    (A^\dag \xi | Y\eta )=\lim\limits_\al (X^\dag _\al \xi | Y\eta ) =
    \lim\limits_\al (\xi | X_\al Y\eta ) = \lim\limits_\al (\xi |
    YX_\al \eta) =(\xi |YA\eta )
\]
for all $\xi ,\eta \in\fd$, it follows that $A\Box Y$ is
well-defined and $A\Box Y=YA $. Furthermore, since
\[
    A\cdot Y=\tau^u _* (\cb)-\lim\limits_\al X_\al Y= \tau_{s^*}
    -\lim\limits_\al X_\al Y,
\]
we have
\[
    (A\cdot Y)\xi =\lim\limits_\al X_\al Y\xi = \lim\limits_\al
    YX_\al \xi = YA\xi =(A\Box Y)\xi
\]
for each $\xi \in \fd$. Hence, $A\cdot Y=A\Box Y$.
\end{proof}

\begin{proposition}\label{pro6.1bis}
 $\cl^{\dag} (\fd, \ch)[\tau_{s^*}]$ is a
locally convex quasi $\cs$-normed algebra over $\cl^{\dag} (\fd )_b \equiv \{ X\in \cl^{\dag} (\fd ) :
\overline{X}\in \cb (\ch)\}$.
\end{proposition}
\begin{proof}  Indeed, as shown in \cite[Section 2.5]{ait3},
$\cl^{\dag} (\fd )_b$, is a $\cs$-normed algebra which is $\tau_{s^*}$ dense in $\cl^{\dag} (\fd,
\ch)$. Hence,  $\cl^{\dag} (\fd,
\ch)$ is a locally convex quasi $\cs$-normed algebra over
$\cl^{\dag} (\fd)_b$.
\end{proof}

\begin{remark}\label{re6.3}
The following questions arise naturally:
\begin{itemize}
    \item[(1)] What is exactly the $\cs$-algebra
 $\cl^\dag (\fd)^\sim _b [\norm_0 ]$?

    Under what conditions  may one have the equality $\cl^\dag (\fd )^\sim _b
    [\norm_0 ]=\cb (\ch )$?
    \item[(2)] Is $\cl^\dag (\fd ,\ch)$ a locally convex quasi $\cs$-algebra
    under the strong$^*$ uniform topology $\tau^u _*$?

    More precisely, does the equality $\cl^\dag (\fd )^\sim _b [\tau^u _* ]
    =\cl^\dag (\fd ,\ch )$ hold?
    \end{itemize}
 We expect the answer to these questions to depend on the properties of the topology
 $t_\dag\equiv t_{\cl^\dag(\fd , \ch)}$ given on $\fd$ and we conjecture positive answers in the case where
    $\fd \equiv\fd^\infty (T)$, with  $T$ a positive self-adjoint operator in a Hilbert space
$\ch$,  and $\norm_0 $ the operator norm
in $\cb(\ch)$. We leave these questions open.
\end{remark}

$\bullet$ In the rest of this Section we consider conditions under
which a locally convex quasi $\cs$-normed algebra is continuously
embedded in a locally convex quasi $\cs$-normed algebra of
operators.

So let $\ca [\tau ]$ be a locally convex quasi $\cs$-normed algebra
over $\ca_0$ and $\fd$ a dense subspace in a Hilbert space $\ch$.
Let $\pi :\ca {\lto} \cl^\dag (\fd,\ch)$ be a $*$-representation. Then we have the
following:

\begin{lemma}\label{le6.4}
Let $\ca [\tau ]$ be a locally convex quasi $\cs$-normed algebra
over $\ca_0$ and $\pi :\ca {\lto} \cl^\dag  (\fd,\ch)$ a $(\tau
,\tau^u _* (\cb))$-continuous $*$-representation of $\ca$. Then,

\emph{(1)} $\pi$ is a  $*$-representation of the $\cs$-algebra
$\widetilde{\ca_0}[ \norm_0 ]$;

    \emph{(2)} $\pi (\ca)[\tau^u _* (\cb)]$ resp. $\pi
(\ca)[\tau_{s^*}]$ are  locally convex
    quasi $\cs$-normed algebras over $\pi(\ca_0 )$.
\end{lemma}

\begin{proof}
(1) Since $\ca_0 \subset \widetilde{\ca_0 }
[\norm_0 ]\subset \ca$ and $\pi$ is a *-representation of $\A$, it follows that
\begin{equation}\label{eq6.2}
\pi(ay)= \pi(a)\mult \pi(y), \quad \forall a \in  \widetilde{\ca_0 }
[\norm_0 ], \forall y \in \ca_0.
\end{equation}
Now we show that
\begin{equation}\label{eq6.3}
\pi(ab)= \pi(a)\mult \pi(b), \quad \forall a,b \in  \widetilde{\ca_0 }
[\norm_0 ].
\end{equation}
Indeed, let $a,b$ be arbitrary elements of $\widetilde{\ca_0 }
[\norm_0 ]$. Then, there exists a sequence $\{y_n\}$ in $\ca_0$ such that
$b=\norm_0-\lim\limits_{n\to\infty} y_n$. Hence, $ab=\norm_0-\lim\limits_{n\to\infty} ay_n$.

Moreover,
it is easily seen that $\pi$ is also $(\tau,\tau_{s^*})$-continuous and so, by \eqref{eq6.2},
\begin{eqnarray*}
    \ip{\pi (b)\xi}{\pi(a^*)\eta} &=&\lim_{n\to \infty} \ip{\pi (y_n )\xi}{\pi(a^*)\eta}
    =\lim_{n\to \infty}\ip{ \pi (a)\mult\pi (y_n )\xi}{\eta} \\
    &=& \lim_{n\to \infty}\ip{ \pi (ay_n )\xi}{\eta}
    =\ip{\pi (ab)\xi}{\eta},
\end{eqnarray*}
for every $\xi, \eta \in  \fd$. Thus, \eqref{eq6.3} holds.\\
For any $\xi\in \fd$, we put
$$f(a)=(\pi(a)\xi|\xi), \quad a\in \widetilde{\ca_0 }
[\norm_0 ].$$
Then, by \eqref{eq6.3}, $f$ is a positive linear functional on the unital
$C^*$-algebra $\widetilde{\ca_0} [\norm_0 ]$. Hence, we have
 $$\|\pi(a)\xi\|^2=f(a^*a)\leq f(\1)\|a\|_0^2=\|\xi\|^2 \|a\|_0^2$$ for all
 $a\in \widetilde{\ca_0} [\norm_0 ]$, which implies that $\pi$ is bounded.
This completes the proof of (1).

(2) $\pi(\ca)$ is a quasi $*$-subalgebra of the locally convex quasi
$*$-algebras $\widetilde{\pi (\ca)}[\tau^u _* (\cb)]$ and
$\widetilde{\pi (\ca)}[\tau_{s^*}]$ over $\pi (\ca_0 )$.
Furthermore, by (1), $\pi (\widetilde{\ca_0} [\norm_0 ])$ is a
$\cs$-algebra and
\[
 \widetilde{\pi ({\ca}_0 )}[\norm_0 ]=\pi (\widetilde{\ca_0} [\norm_0 ])
\subset \pi (\ca).
\]
\end{proof}

\begin{remark}\label{re6.5}
Let $\ca [\tau ]$ be a locally convex quasi $\cs$-normed algebra
over $\ca_0$, and $\pi$ a $(\tau, \tau^u _* (\cb))$-continuous
$*$-representation of $\ca$, where $\cb$ is an admissible subset in
$\cb (\pi (\ca ))$. Let $a \in \ca$ be strongly commuting with $y\in
\widetilde{\ca_0} [\norm_0 ]$. Then $\pi (a)$  commutes strongly
with $\pi (y)$. The converse does not necessarily hold. So
even if $\ca [\tau ]$ is normal, the locally convex quasi
$\cs$-normed algebra $\pi (\ca)$ over $\pi (\ca_0 )$ is not
necessarily normal.
\end{remark}

We are going now to discuss  the faithfulness of a $(\tau
,\tau_{s^*})$-continuous $*$-representation of $\ca$. For this, we
need some facts on sesquilinear forms, for which the reader is
referred to \cite[p. 544]{bmit}. We only recall that if
\[
\mathcal{S}(\ca_0):= \{ \tau \text{-continuous positive invariant
sesquilinear forms } \varphi \text{ on } \ca_0 \times \ca_0 \},
\]
we say that the set $\mathcal{S}(\ca_0)$ is \textit{sufficient},
whenever
\[
a \in \ca \text{ with } \tilde{\varphi}(a,a)=0, \forall \ \varphi
\in \mathcal{S}(\ca_0), \text{ implies } a=0,
\]
where $\tilde{\varphi}$ is the extension of $\varphi$ to a
$\tau$-continuous positive invariant sesquilinear form on $\ca \times
\ca$.

>From the next results, Theorem 6.6 and Corollary 6.7 can be regarded
as generalizations of the analogues of the Gel'fand-Naimark theorem,
in the case of locally convex quasi $C^*$-algebras proved in
\cite[Section 7]{bmit}. Theorem 6.6 is proved in the same way as
\cite[Theorem 7.3]{bmit}.

\begin{theorem}\label{th6.6}
Let $\ca [\tau ]$ be a locally convex quasi $\cs$-normed algebra
over a unital  $\cs$-normed algebra $\ca_0$. The following
statements are equivalent:
\begin{itemize}
\item[(i)] There exists a faithful $(\tau ,\tau_{s^*})$-continuous
$*$-representation of $\ca$.
\item[(ii)] The set $S(\ca_0 )$ is sufficient.
\end{itemize}
\end{theorem}

\begin{corollary}\label{co6.7}
Suppose $S(\ca_0 )$ is sufficient. Then, the locally convex quasi
$\cs$-normed algebra $\ca [\tau ]$ over $\ca_0$ is continuously
embedded in a locally convex quasi $\cs$-normed algebra of
operators.
\end{corollary}

We end this Section with the study of a functional
calculus for the commutatively quasi-positive elements (see
Definition 4.4) of $\ca [\tau]$.

Let $\ca [\tau ]$ be a locally convex quasi $\cs$-normed algebra
over a unital  $\cs$-normed algebra $\ca_0$. If $a \in\ca
[\tau]_{cq+}$,  then by Proposition 4.7(1), the element
$(\mathit{1}+a)^{-1}$ exists and belongs to
$\cu(\widetilde{\ca_0}[\norm_0])$. Denote by $\cs (a)$ the maximal
commutative $\cs$-subalgebra of the $\cs$-algebra $\widetilde{\ca_0}
[\norm_0 ]$ containing the elements $\mathit{1}$ and
$(\mathit{1}+a)^{-1}$.

\begin{lemma}\label{le6.8}
$\widetilde{\cs (a)}[\tau ]$ is a commutative unital locally convex
quasi $\cs$-algebra over $\cs (a)$ and $a\in \widetilde{\cs
(a)}[\tau ]_{q+}$.
\end{lemma}

\begin{proof}
Since $\cs (a)$ is a unital $\cs$-algebra, we have only to check the
properties $\rm (T_1)-(T_5)$. We show $\rm (T_1)$; the rest of them,
as well as the fact that $a\in \widetilde{\cs (a)} [\tau ]_{q+}$ are
proved by the same way as in \cite[Proposition 7.6 and Corollary
7.7]{bmit}. From the condition $\rm (T_3 )$ for ${\ca_0} [\tau ]$, we
have that for all $\la \in \La$, there exist $\la' \in \La$ and
$\gamma_\la >0$ such that
\[
p_\la (xy) \leqq \gamma_\la \| x\|_0 p_{\la'} (y), \forall \ x,y\in
\cs (a).
\]
So, $\cs (a)[\tau ]$ is a locally convex $*$-algebra with separately
continuous multiplication.
\end{proof}

By Lemma \ref{le6.8} and Theorem \ref{th5.2} we can now obtain a
functional calculus for the commutatively quasi-positive elements of
the noncommutative locally convex quasi $\cs$-normed algebra $\ca
[\tau ]$ (see also \cite[Theorem 7.8, Corollary 7.9]{bmit}).

\begin{theorem}\label{th6.9}
Let $\ca[\tau]$ be an arbitrary locally convex quasi $\cs$-normed
algebra over a unital $\cs$-normed algebra $\ca_0$ and $a\in \ca
[\tau]_{cq+}$. Suppose that $a^n$ is well-defined for some $n\in
\N$. Then, there is a unique $*$-isomorphism $f\to f(a)$ from
$\bigcup ^{n} _{k=1} \cc_k (\sigma_{\cs (a)}(a))$ into
$\ca[\tau] \cdot \cs (a)$ such that:
\begin{itemize}
    \item[(i)] If $u_0 (\la )=1$, with $u_0 \in \bigcup ^{n} _{k=1} \cc_k
(\sigma_{\cs (a)}(a))$ and $\la \in \sigma_{\cs (a)}(a)$, then $u_0
(a)=\mathit{1}$.
    \item[(ii)] If $u_1 (\la) =\la$ with $u_1 \in \bigcup ^{n} _{k=1} \cc_k
(\sigma_{\cs (a)}(a))$ and $\la \in \sigma_{\cs (a)}(a)$, then $u_1
(a)=a$.
    \item[(iii)] $(\lambda_1f_1 +f_2 )(a)=\lambda_1f_1 (a)+f_2 (a)$, $\forall$ $f_1
    ,f_2 \in \bigcup ^{n} _{k=1}\cc_k
(\sigma_{\cs (a)}(a))$ and $\la_1 \in \C$;

    $(f_1 f_2 )(a)=f_1 (a)f_2 (a)$, $\forall$ $f_j \in
    \cc_{k_j}(\sigma_{\cs (a)}(a))$, $j=1,2$, with $k_1+k_2 \leqq n$.
    \item[(iv)] The map $f\to f(a)$ restricted to $\cc_b(\sigma_{\cs (a)}(a))$
    is an  isometric $*$-isomorphism of the $\cs$-algebra
    $\cc_b(\sigma_{\cs (a)}(a))$ on the $\cs$-algebra $\cs (a)$.
\end{itemize}
\end{theorem}

Using Theorem \ref{th6.9} and applying Corollary
\ref{co5.3} for the commutative unital locally convex quasi
$\cs$-algebra $\widetilde{\cs (a)}[\tau ]$, we conclude the
following

\begin{corollary}\label{co6.10}
Let $\ca[\tau]$ and $\ca_0$ be as in Theorem \ref{th6.9}. If $a\in
\ca [\tau]_{cq+}$ and $n \in \N$, there is a unique element $b \in
\ca [\tau]_{cq+} \cdot \cs (a)$, such that $a=b^n$. The element $b$
is called commutatively quasi $n$th-root of $a$ and is denoted by
$a^{\frac{1}{n}}$.
\end{corollary}

\section{Applications}

Locally convex quasi C*-normed algebras arise, as we have discussed throughout this paper,
as completions of a C*-normed algebra with respect to a locally convex topology which satisfies
a series of requirements. Completions of
this sort actually occur in quantum statistics.

In statistical physics, in fact, one has to deal with systems
consisting of a very large number of particles, so large that one
usually considers this number to be {\em infinite}. One begins by
considering systems living in a {\em local region} $V$ ($V$ is,
for instance, a bounded region of $\mathbb{R}^3$ for gases or
liquids, or a finite subset of the lattice $\mathbb{Z}^3$ for
crystals) and requires that the set of local regions is directed,
i.e., if $V_1,\,V_2$ are two local regions, then there exists a
third local region $V_3$ containing both $V_1\,\mbox{ and }V_2$.
The observables on a given bounded region $V$ are supposed to
constitute a C*-algebra $\A_V$, where all $\A_V$'s have the same
norm, and so the *-algebra $\A_0$ of local observables,
$\A_0=\bigcup_V \A_V$, is a C*-normed algebra. Its uniform
completion is, obviously, a C*-algebra (more precisely, a quasi
local C*-algebra) that in the original algebraic approach was
taken as the observable algebra of the system. As a matter of
fact, this C*-algebraic formulation reveals to be insufficient,
since for many models there is no way of including in this
framework the thermodynamical limit of the local Heisenberg
dynamics \cite{bag_rev}. Then a possible procedure to follow in
order to circumvent this difficulty is to define in $\A_0$ a new
locally convex topology, $\tau$, called, for obvious reasons, {\em
physical topology}, in such a way that the dynamics in the
thermodynamical limit belongs to the completion of $\A_0$ with
respect to $\tau$. For that purpose, a class of topologies for the
*-algebra $\A_0$ of local observables of a quantum system was
proposed by Lassner in \cite{lass1, lass2}. We will sketch in what
follows this construction.
 Let $\ca_0$ be a
C*-normed algebra to be understood as the algebra of local
observables described above; thus we will suppose that $\ca_0=
\bigcup_{\lambda \in \Lambda}\ca_\lambda$, where
$\{\ca_\lambda\}_\lambda\in \Lambda$ is a family of C*-algebras
labeled  by a directed set of indices $\Lambda$. Assume that, for
every $\alpha \in \Sigma$ ($\Sigma$ a given set of indices),
$\pi_\alpha$ is a $*$-representation of $\ca_0$ on a dense
subspace $\DD_\alpha$ of a Hilbert space $\ch_\alpha$, i.e. each
$\pi_\alpha$ is a $*$-homomorphism of $\ca_0$ into the partial
$O^*$-algebra $\LL^\dag(\DD_\alpha, \ch_\alpha)$ endowed, for
instance, with  the topology $\tau_*^u(\LL^\dag(\DD_\alpha,
\ch_\alpha))$. {\color{red} We shall assume that $\pi_\alpha (x)\DD_\alpha \subset \DD_\alpha$, for every $\alpha \in \Sigma$ and $x \in \ca_0$}. Since every $\ca_\lambda$ is a $C^*$-algebra, each
$\pi_\alpha$ is a bounded {\color{red} and continuous} $*$-representation, i.e.
$\overline{\pi_\alpha(x)}\in \BB(\ch_\alpha)$, {\color{red} $\|\overline{\pi_\alpha(x)}\|\leq \|x\|_0$,} for every $x \in
\ca_0$. {\color{red} So each $\pi_\alpha$ can be extended to the $C^*$-algebra $\widetilde{\ca_0} [\|\cdot\|_0]$ (we denote the extension by the same symbol)}.The family is supposed to be {\em faithful}, in the sense
that if {\color{red} $x \in \widetilde{\ca_0} [\|\cdot\|_0]$}, $x \neq 0$, then there exists $\alpha \in
\Sigma$ such that $\pi_\alpha(x)\neq 0$. Let us further suppose
that $\DD_\alpha= \DD^\infty(M_\alpha)=\bigcap_{n \in
\N}\DD(M_\alpha^n)$, where $M_\alpha$ is a selfadjoint operator.
Without loss of generality we may assume that $M_\alpha\geq
I_\alpha$, with $I_\alpha$ the identity operator in
$\BB(\ch_\alpha)$. Under these assumptions, a {\em physical}
topology $\tau$ can be defined on $\A_0$ by the family of
seminorms
$$ p_\alpha^{f} (x)=  \|\pi_\alpha(x)
f(M_\alpha)\|+ \|\pi_\alpha(x^*) f(M_\alpha)\|, \ x \in \A_0,$$
where $\alpha \in \Sigma$ and $f$ runs over the set ${\mathcal F}$
of all positive, bounded and continuous functions $f(t)$ on
${\mathbb R}^+$ such that
$$ \sup_{t \in {\mathbb R}^+} t^kf(t)<\infty, \quad \forall \
k=0,1,2,\ldots.$$ Then, $\ca_0[\tau]$ is a locally convex
$*$-algebra with separately continuous multiplication (i.e.
(T$_1$) holds). In order to prove that $\widetilde{\A_0}[\tau]$ is
a locally convex quasi $C^*$-normed algebra, we need to prove
that {\color{red}(T$_2$)-(T$_5$) also} hold. As for
(T$_2$), we have, for every $\alpha \in \Sigma$,
$$p_\alpha^{f} (x)= \|\pi_\alpha(x)
f(M_\alpha)\|+ \|\pi_\alpha(x^*) f(M_\alpha)\| \leq
2\|f(M_\alpha)\| \|\pi_\alpha(x)\| \leq 2\|f(M_\alpha)\|\|x\|_0, \
 x \in \ca_0.$$

The compatibility of $\tau$ with $\|\cdot\|_0$ follows easily from the closedness of
the operators $f(M_\alpha)^{-1}$ and
the faithfulness of the family $\{\pi_\alpha\}_{\alpha \in \Sigma}$ of *-representations.

The condition (R) does not hold, in general, but, on the other hand,
if $x,y \in \A_0$ with $xy=yx$, we have
\begin{eqnarray*}p_\alpha^{f} (xy)&=& \|\pi_\alpha(xy)
f(M_\alpha)\|+ \|\pi_\alpha((xy)^*)
f(M_\alpha)\|\\&=&\|\pi_\alpha(xy) f(M_\alpha)\|+
\|\pi_\alpha(x^*y^*)f(M_\alpha)\|\\ &\leq & \|\pi_\alpha(x)\|(
\|\pi_\alpha(y) f(M_\alpha)\|+ \|\pi_\alpha(y^*)f(M_\alpha)\| ) \\
&=& \|\pi_\alpha(x)\|p_\alpha^{f} (y) \leq  \|x\|_0 p_\alpha^{f}
(y).\end{eqnarray*} Hence (T$_3$) holds.
{\color{red} As for (T$_4$), we begin with noticing that for every $\alpha \in \Sigma$, $\pi_\alpha(\A_0)$ is an O*-algebra of bounded operators in $\DD_\alpha$. Hence, its closure in $\LL^\dag(\DD_\alpha, \ch_\alpha)[\tau_*^u(\LL^\dag(\DD_\alpha,
\ch_\alpha))]$ is a locally convex C*-normed algebra of operators, by Proposition \ref{pro6.1}. Moreover, every $\pi_\alpha$ can be extended by continuity to $\widetilde{\ca_0} [\|\cdot\|_0]$. The extension, that we denote by the same symbol, takes values in $\LL^\dag(\DD_\alpha, \ch_\alpha)[\tau_*^u(\LL^\dag(\DD_\alpha,
\ch_\alpha))]$, since this space is complete. Now, if $\{x_\lambda\}$ is a net in $\cu (\widetilde{\ca_0} [\norm_0 ])_+$ $\tau$-converging to $x \in \widetilde{\ca_0} [\norm_0 ])$, then $x=x^*$ and $\pi_\alpha(x_\lambda)\to \pi_\alpha(x)$ in $\LL^\dag(\DD_\alpha, \ch_\alpha)[\tau_*^u(\LL^\dag(\DD_\alpha,
\ch_\alpha))]$, for every $\alpha \in \Sigma$. Thus $\pi_\alpha(x)\geq 0$ and $\|\pi_\alpha(x)\|\leq 1$, for every $\alpha\in \Sigma$, since the same is true for every $x_\lambda$. By constructing a faithful representation $\pi$ by direct sum of the $\pi_\alpha$'s, one easily realizes that $x \geq 0$ and $\|x\|_0\leq 1$. The inclusion $\widetilde{\ca_0} [\tau]_{q+}\cap \widetilde{\ca_0}
    [\norm_0 ]\subset \widetilde{\ca_0} [\norm_0 ]_+$ in Condition (T$_5$) can be proved in similar fashion. The converse inclusion comes from Lemma \ref{le4.5}. Thus Condition (T$_5$) holds.

Then we conclude that
\begin{stat}$\ca\equiv
\widetilde{\ca_0} [\tau]$ is a locally convex quasi C*-normed
algebra, which can be understood as the quasi *-algebra of the
observables of the physical system. \end{stat}
}
\bigskip A more concrete realization of the situation discussed above is obtained for the so-called BCS model.
Let $V$ be a finite region of a $d$-dimensional lattice $\Lambda$ and $| V |$ the number
of points in $V$. The local $C^*$-algebra ${\A}_V$ is generated by the Pauli operators $\vec\sigma_p = (\sigma^1_p,
\sigma^2_p, \sigma^3_p)$ and by the unit $2\times 2$ matrix $e_p$ at every point
 $p \in V$. The $\vec\sigma_p$'s are copies of the Pauli matrices localized in $p$.

If $V \subset V^{'}$ and $A_V \in {\A}_V$, then $A_V \rightarrow A_{V^{'}} = A_V \otimes ({{\atop \bigotimes }
\atop{p \in V^{'} \setminus V}} e_p)$ defines the natural imbedding of ${\A}_V$ into ${\A}_{V^{'}}$.

Let ${\vec n}=(n_1, n_2, n_3)$ be a unit vector in ${{\mathbb R}}^3$, and put $ (\vec\sigma\cdot {\vec n}) = n_1 \sigma^1
 + n_2 \sigma^2 + n_3 \sigma^3.
$ Then, denoting as $Sp(\vec\sigma \cdot {\vec n})$ the spectrum of $\vec\sigma \cdot {\vec n}$,
we have $ Sp(\vec\sigma \cdot {\vec n}) = \{ 1,
-1\}. $ Let $|\vec n \rangle\in {\mathbb C}^2$ be a unit eigenvector associated with $1$.

Let now denote by ${\mathbf n}:= \{ \vec n_p \}_{p\in\Lambda}$  an
infinite sequence of unit vectors in ${{\mathbb R}}^3$  and $| \mathbf{n} \rangle = {{\atop \bigotimes }\atop{p}} |
\vec n_p\rangle$
 the corresponding unit vector in
the infinite tensor product ${\cal H}_\infty = {{\atop \bigotimes }\atop{p}}
 {\mathbb C}_p^2$.
We put $ {\A}_0 = \bigcup_V {\A}_V $ and $ {\cal D}^0_{\mbox{\small$\bfn$}} = {\A}_0 |\bfn\rangle $
and we denote the closure
of ${\cal D}^0_{\mbox{\small$\bfn$}}$ in ${\cal H}_\infty$ by
 $ {\cal H}_{\mbox{\small$\bfn$}}$.
As we saw above, to any sequence $ \bfn$  of three-vectors there corresponds a state $|\bfn\rangle$ of the
system. Such a state defines a realization $\pi_{\mbox{\small$\bfn$}}$ of ${\A}_0$
in the Hilbert space ${\cal H}_{\mbox{\small$\bfn$}}$. This
representation is faithful, since the norm completion ${\A}_S$ of ${\A}_0$ is a  simple C*-algebra.
A special basis for $ {\cal H}_{\mbox{\small$\bfn$}}$ is obtained from
the {\em ground} state $|\bfn\rangle$ by {\em flipping} a finite number of spins using the following strategy:\\
Let $\vec n$ be a unit vector in ${{\mathbb R}}^3$, as above, and $ |\vec n\rangle$ the corresponding vector
of ${\mathbb C}^2$. Let us choose two other unit vectors ${\vec n}^1,
{\vec n}^2$ so that $({\vec n}, {\vec n}^1, {\vec n}^2)$
form an orthonormal basis of ${{\mathbb R}}^3$. We put $ {\vec n}_{\pm} = \frac{1}{2} ({\vec n}^1 \pm i{\vec n}^2) $
and define
$ |m,\vec n\rangle := (\vec\sigma \cdot {\vec n}_{-})^m |\vec n\rangle \ \ (m=0,1). $ Then we have
$$
(\vec\sigma \cdot {\vec n}) |m, \vec n\rangle = (-1)^m |m,\vec n\rangle \ \ (m=0,1).
$$
Thus, the set $ \left \{|\mathbf{m}, \mathbf{n}\rangle = {{\atop
\bigotimes }\atop{p}} |m_p, \vec{n}_p\rangle ;\ m_p = 0, 1,\ \
{\displaystyle \sum_p} m_p < \infty \right \}$ forms an orthonormal
basis in ${\cal H}_{\mbox{\small$\mathbf{n}$}}$.

In
this space we define the unbounded self-adjoint operator ${M}_{\mbox{\small$\mathbf{n}$}}$
 by
\begin{equation}
{M}_{\mbox{\small$\mathbf{n}$}} |\mathbf{m}, \mathbf{n}\rangle= (1+\sum_p m_p)|\mathbf{m}, \mathbf{n}\rangle.
\label{opem}
\end{equation}
${M}_{\mbox{\small$\mathbf{n}$}}$ counts the number of the flipped
 spins in $|\mathbf{m}, \mathbf{n}\rangle$ with
respect to the ground state $| \mathbf{n} \rangle$. Now we put
$$
{\DD}_{\mbox{\small$\mathbf{n}$}} = \bigcap_k
{\DD}({M}_{\mbox{\small$\mathbf{n}$}}^k),
$$
 The representation $\pi_{\mbox{\small$\mathbf{n}$}}$ is defined on the basis vectors
 $\{ |\mathbf{m}, \mathbf{n}\rangle \}$ by
 $$
\pi_{\mbox{\small$\mathbf{n}$}} (\sigma^i_p)|\mathbf{m}, \mathbf{n}\rangle= \sigma^i_p \mid m_p, \vec n_p\rangle
  \otimes ({{\atop
\prod }\atop{\scriptstyle{p^{'} \neq p}}} \otimes \mid m_{p ^{'}}, \vec n_{p^{'}}\rangle)\ \ \ (i= 1, 2, 3).
$$
This definition is then extended in obvious way to the whole space
${\cal H}_{\mbox{\small$\mathbf{n}$}}$. It turns out that
$\pi_{\mbox{\small$\mathbf{n}$}}$ is a {\em bounded} representation
of $\A_0$ into ${\cal H}_{\mbox{\small$\mathbf{n}$}}$. For more
details we refer to \cite{trapani2,bagtra3}. Hence, the procedure
outlined above applies, showing that a natural framework for
discussing the BCS model is, indeed, provided by locally convex
quasi C*-normed algebras considered in this paper. We argue that an
analysis similar to that of \cite{bagtra3} can be carried out also
in the present context, so that for suitable finite volume
hamiltonians, the thermodynamical limit of the local dynamics can be
appropriately defined in $\widetilde{\ca_0} [\tau]$.


\begin{flushleft}
Dipartimento di Matematica ed Applicazioni, Fac Ingegneria,
Universita di\\ Palermo, I-90128 Palermo, Italy\\
E-mail address: bagarell@unipa.it
\end{flushleft}

\begin{flushleft}
Department of Mathematics, University of Athens,
Panepistimiopolis,\\ Athens 15784, Greece\\ E-mail address:
fragoulop@math.uoa.gr
\end{flushleft}

\begin{flushleft}
Department of Applied Mathematics, Fukuoka University,\\ Fukuoka 814-0180, Japan\\
E-mail address: a-inoue@fukuoka-u.ac.jp
\end{flushleft}

\begin{flushleft}
Dipartimento di Matematica ed Applicazioni, Universita
di Palermo,\\ I-90123 Palermo, Italy\\
E-mail address: trapani@unipa.it
\end{flushleft}

\end{document}